\documentclass[10pt,twocolumn,twoside]{IEEEtran}%,draft

\usepackage{amsmath,amssymb,amsthm}
\usepackage{algorithmic}
\usepackage{amsmath,graphicx}
\usepackage{pgf,tikz,pgfplots}
\usepackage{bm}
\usepackage{color}
\usepackage{float}

\usepackage{subfigure}
\usepackage{cases}
\usepackage{tabularx}

\usepackage{multicol}
\usepackage{todonotes}

\usepackage{hyperref}
\hypersetup{
	colorlinks=true,
	linkcolor=blue,
	filecolor=magenta,
	urlcolor=cyan,
	citecolor=blue
}
\usetikzlibrary{arrows,decorations.markings,automata}
\usepackage{xspace}

\newtheorem{theorem}{Theorem}

\newtheorem{lemma}{Lemma}

\newtheorem{corollary}{Corollary}

\newtheorem{example}{Example}

\newcommand{\rev}[1]{{\color{black}{#1}}}
\newcommand{\revv}[1]{{\color{black}{#1}}}
\newcommand{\revvv}[1]{{\color{black}{#1}}}
\newcommand{\LPRR}{\text{LPdRR}\xspace}
\newcommand{\LPRRS}{\text{LP'dRR}\xspace}
\newcommand{\LPRRSS}{\text{LPsRR}\xspace}
\newcommand{\GREEDY}{\text{IO}\xspace}
\newcommand{\LPR}{\text{LPoR}\xspace}
\newcommand{\PSO}{\text{PSO}\xspace}
\newcommand{\EXACT}{\text{EXACT}\xspace}

\definecolor{ududff}{rgb}{0.30196078431372547,0.30196078431372547,1.}
\usepackage{algorithm,algorithmic}
\makeatletter
\newcommand\Label[1]{&\refstepcounter{equation}(\theequation)\ltx@label{#1}&}
\makeatother

\newcommand{\figuresize}{1.8in}

\begin{document}
%	\setlength{\abovedisplayskip}{0.133cm}
%	\setlength{\belowdisplayskip}{0.133cm}
	%\setlength{\textfloatsep}{\baselineskip plus 0.2\baselineskip minus 0.2\baselineskip}
	%\setlength{\belowcaptionskip}{-10pt}
%	\setlength{\jot}{0.09cm}
%An Efficient LP Rounding-and-Refinement Algorithm for Large-Scale Network Slicing
\title{Towards Efficient Large-Scale Network Slicing: An LP \revv{Dynamic} Rounding-and-Refinement Approach\thanks{Part of this work \cite{Chen2021} has been presented at the 2021 IEEE International Conference on Acoustics, Speech and Signal Processing  (ICASSP), Toronto, Ontario, Canada, June 6--11, 2021.}}
\author{\IEEEauthorblockN{Wei-Kun Chen, Ya-Feng Liu, Fan Liu, Yu-Hong Dai, and Zhi-Quan Luo}% <-this % stops a space
	\thanks{The work of W.-K. Chen was supported in part by National Natural Science Foundation of China (NSFC) under Grant 12101048. 
		The work of Y.-F. Liu was supported in part by the NSFC under Grant 12288201 and Grant 12022116.
		The work of F. Liu was supported in part by the NSFC under Grant 62101234.
		%The work of Y.-H. Dai was supported in part by the NSFC  under Grant 12021001, Grant 11991021, Grant 11991020, and Grant 11971372, the National Key R\&D Program of China under Grant 2021YFA1000300 and Grant 2021YFA1000301, and the Strategic Priority Research Program of Chinese Academy of Sciences under Grant No. XDA27000000.
		The work of Y.-H. Dai was supported in part by the NSFC  under Grant 12021001, 11991021, 11991020, and 11971372.
%		The work of Y.-H. Dai was supported in part by the NSFC under Grant  12021001, Grant 11991021, Grant 11991020, and Grant 11971372, the National Key  R\&D Program of China under Grant 2021YFA1000300 and Grant 2021YFA1000301, and the Strategic Priority Research Program of Chinese Academy of Sciences under Grant XDA27000000.
		The work of Z.-Q. Luo was supported in part by the National Key Research and Development Project under Grant 2022YFA1003900, and in part by the Guangdong Provincial Key Laboratory of Big Data Computing.
		(\emph{Corresponding author: Ya-Feng Liu.})}
	\thanks{W.-K. Chen is with the School of Mathematics and Statistics/Beijing Key Laboratory on MCAACI, Beijing Institute of Technology, Beijing 100081, China (e-mail: chenweikun@bit.edu.cn).
		Y.-F. Liu {and Y.-H. Dai} are with the State Key Laboratory
		of Scientific and Engineering Computing, Institute of Computational Mathematics and Scientific/Engineering Computing, Academy of Mathematics and Systems Science, Chinese Academy of Sciences, Beijing 100190, China (e-mail: \{yafliu, dyh\}@lsec.cc.ac.cn).
		F. Liu is with the Department of Electrical and Electronic Engineering, Southern University of Science and Technology, Shenzhen 518055, China (e-mail: liuf6@sustech.edu.cn).
		Z.-Q. Luo is with the Shenzhen Research Institute of Big Data and The Chinese University of Hong Kong, Shenzhen 518172, China (e-mail: luozq@cuhk.edu.cn)
	}
}

\maketitle

\begin{abstract}
	In this paper, we propose an efficient algorithm for the network slicing problem which attempts to map multiple customized virtual network requests (also called services) to a common shared network infrastructure and allocate network resources to meet diverse service requirements.
	The problem has been formulated as a mixed integer linear programming (MILP) formulation in the literature. 
	\rev{We} first propose a novel linear programming (LP) relaxation of the MILP formulation.
	We show that compared with a natural LP relaxation of the MILP formulation, the novel LP relaxation is much more compact in terms of smaller numbers of variables and constraints, and much stronger in terms of providing a better LP bound, which makes it particularly suitable to be embedded in an LP \revv{relaxation} based algorithm.
	Then we design an efficient two-stage LP \revv{dynamic} rounding-and-refinement algorithm based on this novel LP relaxation.
	In the first stage, the proposed algorithm uses \revv{an} LP \revv{dynamic} rounding procedure to place the virtual network functions of all services into cloud nodes while taking traffic routing of all services into consideration;
	in the second stage, the proposed algorithm uses an iterative LP refinement procedure to obtain a solution for traffic routing of all services with their end-to-end delay constraints being satisfied.
	Compared with the existing algorithms which either have an exponential complexity or return a low-quality solution, our proposed algorithm achieves a better trade-off between the solution quality and the computational complexity.
	In particular, the worst-case complexity of our proposed algorithm is polynomial, which makes it suitable for solving large-scale problems. 
	Numerical results demonstrate the effectiveness and efficiency of our proposed algorithm.
\end{abstract}
\begin{IEEEkeywords}
	LP Relaxation, Network Slicing, Resource Allocation, Rounding-and-Refinement.
\end{IEEEkeywords}

\section{Introduction}
Network function virtualization (NFV) plays a crucial role in the fifth generation (5G) and beyond 5G networks \cite{Mijumbi2016}.
Different from traditional networks where service functions are processed by specialized hardwares in fixed locations, NFV efficiently takes the advantage of cloud technologies to configure some specific nodes (called cloud nodes) in the network to process network service functions on-demand, and then flexibly establishes a customized virtual network for each service request.
However, as virtual network functions (VNFs) of all services run over a shared common network infrastructure, it is crucial to allocate network (e.g., cloud and communication) resources to meet the diverse service requirements.
The above resource allocation problem in the NFV-enabled network is called \emph{network slicing} in the literature.

\subsection{Related Works}
Various approaches have been proposed to solve the network slicing problem or its variants; see \cite{Chen2020}--\cite{Promwongsa2020} and the references therein.
These approaches can generally be classified into two categories: (i) exact algorithms that solve the problem to global optimality and (ii) heuristic algorithms that aim to quickly find a feasible solution for the problem.

More specifically, references \cite{Chen2020}--\cite{Addis2015} proposed the link-based mixed integer linear programming (MILP) formulations for the network slicing problem and used standard MILP solvers like Gurobi \cite{Gurobi} to solve their problem formulations.
Reference \cite{Ayoubi2019} used a logic-based Benders decomposition approach  \cite{Hooker2003} to solve the link-based MILP formulation.
References \cite{Hu2013}--\cite{Liu2017} proposed the path-based MILP formulations and used a column generation approach \cite{Conforti2014} to solve the related problems.
Though the above three approaches can solve the network slicing problem to global optimality, they generally suffer from low computational efficiency as their worst-case complexities are exponential.

Various heuristic and meta-heuristic algorithms has also been proposed to quickly obtain a feasible solution of the network slicing problem; see  \cite{Zhang2017}--\cite{Promwongsa2020}.
For instance, references \cite{Zhang2017}-\cite{Luizelli2015} simplified the solution approach by decomposing the network slicing problem into a VNF placement  subproblem (which maps VNFs into cloud nodes in the network) and a traffic routing subproblem (which finds paths connecting two adjacent VNFs in the network) and solving each subproblem separately.
To obtain a binary solution for the VNF placement subproblem, references \cite{Zhang2017}--\cite{Chowdhury2012} first solved the linear programming (LP) relaxation of the network slicing problem and then used a \revv{one-shot} rounding strategy while references \cite{Yu2008}--\cite{Luizelli2015} used some greedy heuristics (without solving any LP).
Once the VNFs are mapped \revv{into} the cloud nodes, the traffic routing subproblem is solved by using shortest path, $k$-shortest path, or multicommodity flow algorithms.
However, solving the VNF placement subproblem without taking the global information (i.e., traffic routing of all services) into account can lead to infeasibility or low-quality solutions.
Reference \cite{Mohammadkhan2015} decomposed the whole problem into several subproblems by dividing the services into several groups and iteratively solved the subproblem for each group using standard MILP solvers.
Similarly, this approach failed to take the global information (i.e., VNF placement and traffic routing of all services) into consideration and hence can also lead to infeasibility or low-quality solutions.
References \cite{Zhang2012}, \cite{Li2015}, and \cite{Abu-Lebdeh2017}--\cite{Promwongsa2020} proposed \revv{the particle swarm optimization (PSO)}, simulated annealing based, and Tabu search based meta-heuristic algorithms, respectively, to obtain a feasible solution for the network slicing problem.

To summarize, the aforementioned approaches to solve the network slicing problem either have an exponential complexity (e.g., \cite{Chen2020}--\cite{Liu2017}) or return a low-quality solution (e.g., \cite{Zhang2017}--\cite{Promwongsa2020}).
The motivation of this paper is to fill this research gap, i.e., develop an algorithm that finds a high-quality solution of the network slicing problem while still enjoys a polynomial time computational complexity.
\rev{Integer programming and LP relaxation techniques play vital roles in the developed algorithm.}

\subsection{Our Contributions}

The main contribution of this paper is the proposed algorithm, that achieves a good trade-off between high solution quality and low computational complexity, for solving the network slicing problem.
In particular,
\begin{itemize}
	\item
	      \rev{We} first propose a novel LP relaxation of the MILP problem formulation.
	      We show that compared with the natural LP relaxation of the problem (which directly relaxes the binary variables into continuous variables in $[0,1]$), the novel LP relaxation enjoys two key advantages: (i) it is much more compact in terms of smaller numbers of variables and constraints; (ii) it is much stronger in terms of providing a better LP bound, which is crucial to the effectiveness of the proposed algorithm.
	      \rev{We also analyze the integrality gap of the proposed LP relaxation.}
	      %We show that in some special case, this gap can be infinite, which provides more insight into the intrinsic difficulty of the original problem. In addition, we also identify two other special cases where we can derive bounded or zero integrality gaps.
	\item We then develop a two-stage LP \revv{dynamic} rounding-and-refinement algorithm (\LPRR) based on the above novel LP relaxation.
	      Specifically, in the first stage, we solve the VNF placement subproblem by using an LP \revv{dynamic} rounding procedure, which takes traffic routing into account;
	      in the second stage, we solve the traffic routing subproblem by using an iterative LP refinement procedure to find a solution that satisfies the end-to-end (E2E) delay constraints of all services.
	      The proposed algorithm has a guaranteed polynomial time worst-case complexity, and thus is particularly suitable for solving large-scale problems.
	      \revv{We also analyze the approximation gap of the proposed \LPRR algorithm.}
\end{itemize}
Simulation results demonstrate the effectiveness and efficiency of our proposed LP relaxation and algorithm.
To be more specific, our simulation results demonstrate that (i)
the novel LP relaxation  significantly outperforms the natural LP relaxation in terms of the solution efficiency and providing a better LP bound;
(ii) our proposed \LPRR algorithm is more effective than the existing state-of-the-art algorithms \rev{in \cite{Chen2020},  \cite{Chowdhury2012},  \cite{Yu2008}}, \revv{and \cite{Zhang2012}}  in terms of both solution efficiency and quality.

The rest of the paper is organized as follows.
Section \ref{sec:modelformulation} briefly reviews the network slicing problem and its mathematical formulation.
Sections \ref{novellp} and \ref{sec:heur} propose the novel LP relaxation and the \LPRR algorithm for the network slicing problem, respectively.
Section \ref{sec_num} reports the computational results.
Finally, Section \ref{sec:conclusion} draws the conclusion.
\section{System model and problem formulation}
\label{sec:modelformulation}
\subsection{System Model}
Let $\mathcal{G}=\{\mathcal{I},\mathcal{L}\}$ be the connected directed network, where $\mathcal{I}=\{i\}$ and $\mathcal{L}=\{(i,j)\}$ are the sets of nodes and links, respectively.
Each link $ (i,j) $ has an expected (communication) delay $ d_{ij} $ \cite{Woldeyohannes2018,Luizelli2015,Mohammadkhan2015}, and a total data rate upper bounded by the capacity $C_{ij}$.
The set of cloud nodes is denoted as $ \mathcal{V} \subseteq \mathcal{I}$.
Each cloud node $ v $ has a computational capacity $ \mu_v $ and processing one unit of data rate requires one unit of (normalized) computational capacity, as assumed in \cite{Zhang2017}.
A set of flows $\mathcal{K}=\{k\}$ is required to be supported by the network.
The source and destination nodes of flow $k$ are denoted as $S(k)$ and $D(k)$, respectively, with $S(k),D(k)\notin \mathcal{V}$.
Each flow $ k $ relates to a customized service, which is given by a service function chain (SFC) consisting of $ \ell_k $ service functions that have to be processed in sequence by the network: $f_{1}^k\rightarrow f_{2}^k\rightarrow \cdots \rightarrow f_{\ell_k}^k$ \cite{Zhang2013,Halpern2015,Mirjalily2018}.
%\begin{equation}
%\label{sequence}
%f_{1}^k\rightarrow f_{2}^k\rightarrow \cdots \rightarrow f_{\ell_k}^k.
%\end{equation}
%
To minimize the coordination overhead, each function must be processed at exactly one cloud node, as required in \cite{Domenico2020,Zhang2017,Woldeyohannes2018}.
If function $ f^k_s $, $ s \in \mathcal{F}(k) := \{1,\ldots, \ell_k\} $, is processed by
cloud node $ v $ in $ \mathcal{V} $, the expected NFV delay is assumed to be known as $
	d_{v,s}(k) $, which includes both processing and queuing delays \cite{Woldeyohannes2018,Luizelli2015}.
For flow $ k $, the service function rates before receiving any function and after receiving {function} $ f^k_s $ are denoted as $ \lambda_0(k) $ and $ \lambda_s(k) $, respectively.
Each flow $ k $ has an E2E delay requirement, denoted as $ \Theta_k $.

The network slicing problem is to determine functional instantiation, the routes, and the associated data rates on the corresponding routes of all flows while satisfying the capacity constraints on all cloud nodes and links, the SFC requirements, and the E2E delay requirements of all flows.
Next, we shall introduce the problem formulation in details.
\subsection{Problem Formulation}
%{\bf\noindent$\bullet$ VNF Placement\vspace{0.1cm}\\}
\subsubsection{VNF Placement}
%\label{vnf}
We introduce the binary variable $x_{v,s}(k)$ to indicate whether or not function $f^k_s$ is processed by cloud node $v$.
\rev{Notice that in practice, cloud node $v$ may not be able to process function $f_s^k$ \cite{Chen2020,Zhang2017}, and in this case, we can simply set $x_{v,s}(k) = 0$.} 
Each function $f_s^k$ must be processed by exactly one cloud node, i.e.,
\begin{eqnarray}
	\label{onlyonenode}
	\sum_{v\in \mathcal{V}}x_{v,s}(k)=1,~\forall ~k \in \mathcal{K},~ \forall ~s\in  \mathcal{F}(k).
\end{eqnarray}
Let $y_v=1$ denote that cloud node $v$ is activated and powered on; otherwise $y_v=0$. Thus
\begin{equation}
	\label{xyrelation}
	x_{v,s}(k) \leq  y_v, ~ \forall~v \in \mathcal{V},~\forall~k \in \mathcal{K},~\forall~s \in \mathcal{F}(k).
\end{equation}
The node capacity constraints can be written as follows:
\begin{equation}
	\label{nodecapcons}
	\sum_{k\in \mathcal{K}}\sum_{s \in \mathcal{F}(k)}\lambda_s(k)x_{v,s}(k)\leq \mu_v y_v,~\forall~ v \in \mathcal{V}.
\end{equation}
%
%{\bf\noindent$\bullet$ Traffic Routing\vspace{0.1cm}\\}
\subsubsection{Traffic Routing}
 Let $ (k,s) $ denote the flow which is routed between the two cloud nodes hosting two adjacent functions $ f_s^k $ and $ f_{s+1}^k $.
Similar to \cite{Chen2020}, we suppose that there are at most $P$ paths that can be used to route flow $(k,s)$ and denote $\mathcal{P}=\{1, \ldots, P\}$.
\rev{The main reason for introducing notation $P$ is to model the delay of flow $(k,s)$ and  allow the corresponding traffic flow to be split into multiple paths.
	In practice, flow $(k,s)$ can be divided into an arbitrary number of paths, meaning that $P$ can be set as a sufficiently large integer.
	It is worthwhile highlighting that setting  $P=|\mathcal{L}|$ is sufficient for flow $(k,s)$ to be divided into an arbitrary number of paths. 
	\revv{This} is because, from the classical network flow theory, any traffic flow between two nodes can be decomposed into the sum of at most $|\mathcal{L}|$ routes on the paths and a circulation; see \cite[Theorem 3.5]{Ahuja1993}.}

Let $ r(k,s,p) $ be the fraction of data rate $\lambda_s(k)$ on the $ p $-th path of flow $ (k,s) $.
Then, the following constraint enforces that the total data rate between the two nodes hosting functions $ f_s^k $ and  $ f_{s+1}^k $ is equal to $ \lambda_s(k) $:
\begin{align}
	  & \sum_{p \in \mathcal{P}}  r(k, {s}, p) =  1,   ~ \forall~ k \in \mathcal{K},~\forall~s\in \mathcal{F}(k)\cup \{0\} \label{relalambdaandx11}.
\end{align}

Let $ z_{ij}(k,s,p)\in \{0,1\}$ denote whether or not link $ (i,j) $ is on the $ p $-th path of flow $ (k,s) $ and $ r_{ij}(k,s,p) $ be the associated fraction of data rate $\lambda_{s}(k)$. Then
\begin{align}
	  & r_{ij}(k, s,  p ) = r(k,s,p) z_{ij}(k, s,p ), \nonumber                                                                                                  \\
	  & ~~\forall~(i,j) \in {\mathcal{L}}, ~\forall~k \in \mathcal{K}, ~\forall~s \in \mathcal{F}(k)\cup \{0\},~\forall~p \in \mathcal{P}. \label{nonlinearcons}
\end{align}
The total data rates on link $ (i,j) $ is upper bounded by capacity $ C_{ij} $:
\begin{equation}
	\label{linkcapcons1}
	\sum_{k \in \mathcal{K}} \sum_{s\in \mathcal{F}(k) \cup \{0\}}\sum_{p \in \mathcal{P}} \lambda_{s}(k) r_{ij}(k, s,p) \leq C_{ij}, ~  \forall~(i,j) \in \mathcal{L} .
\end{equation}
\subsubsection{SFC}
To ensure that the functions of each flow $k$ are processed in the prespecified order $f_{1}^k\rightarrow f_{2}^k\rightarrow \cdots \rightarrow f_{\ell_k}^k$ and for each $s \in \mathcal{F}(k)\cup \{0\}$ and $p \in \mathcal{P}$, $\{(i,j) \mid z_{ij}(k,s,p)=1 \}$ forms a path, we need the flow conservation constraint \eqref{mediacons2} in the next page.
We remark that together with constraint \eqref{nonlinearcons}, constraint \eqref{mediacons2} also ensures that $\{(i,j) \mid r_{ij}(k,s,p) >0 \}$ forms a path.
Notice that in \eqref{mediacons2}, the term $x_{i,s+1}(k) - x_{i,s}(k)$ takes values $1$, $0$, or $-1$ depending on whether (i) node $i$ hosts function $f_{s+1}^k$ but does not host function $f_{s}^k$; (ii) node $i$ hosts \revv{both} functions $f_s^k$ and $f_{s+1}^k$ or does not host any of the two functions; or (iii) node $i$ hosts function $f_s^k$ but does not host function $f_{s+1}^k$.
In constraint \eqref{mediacons2}, when $s = 0$ and $s=\ell_k +1$, we let
\begin{equation*}
	\begin{aligned}
		x_{i,s}(k) & = & \!\!\!\!\!\left\{\begin{array}{ll}1,
			~ {\text{if}}~ i=S(k);~ \\
			0, ~{\text{otherwise}},\end{array}\right.\!\!\!\!\!\!
		{\text{and}}~
		x_{i,s}(k) & = & \!\!\!\!\!\left\{\begin{array}{ll}1,
			~ {\text{if}}~ i=D(k);~  \\
			0, ~ {\text{otherwise}},\end{array}\right.
	\end{aligned}
\end{equation*}
respectively.

\begin{figure*}[t]
	\begin{equation}
		\sum_{j: (j,i) \in \mathcal{{L}}} z_{ji}(k, s, p) - \sum_{j: (i,j) \in \mathcal{{L}}} z_{ij}(k, s,  p)=\left\{\begin{array}{ll}
			0,                        & \text{if}~ i   \in {\mathcal{I}}\backslash {\mathcal{V}}; \\
			x_{i,s+1}(k)- x_{i,s}(k), & \text{if}~  i \in {\mathcal{V}},
		\end{array} \right.	\forall~k \in \mathcal{K},~\forall~s \in \mathcal{F}(k)\cup \{0\},~ \forall~ p \in \mathcal{P} \label{mediacons2}.  
	\end{equation}
	\vspace{-5mm}
	 \hrulefill 
	%\vspace*{-4pt}
\end{figure*}
\subsubsection{E2E Delay}
Let $ \theta(k,s) $ denote the communication delay due to the traffic flow from the cloud node hosting function $ f^k_s $ to the cloud node hosting function $ f^k_{s+1} $. 
\revv{
By definition, $\theta(k,s)$ must be the largest delay among the $P$ paths of flow $(k,s)$, i.e.,
\begin{align}
	& \theta(k,s) = \max_{p \in  \mathcal{P}}\left\{\sum_{(i,j) \in \mathcal{{L}}}  d_{ij}  z_{ij}(k, s, p)\right\},\nonumber\\
	& \qquad \qquad\qquad\qquad\qquad\forall~k \in \mathcal{K}, ~s\in \mathcal{F}(k)\cup \{0\}. 	\label{maxdelay1}
\end{align}}%
To ensure that flow $k$'s {E2E} delay is less than or equal to its threshold $\Theta_k$, we need the following constraint:
\begin{equation}
	\label{delayconstraint}
	\theta_N(k) +\theta_L(k)  \leq \Theta_k,~\forall~k \in  \mathcal{K},
\end{equation}
where  $\theta_N(k) =  \sum_{v \in \mathcal{{V}}}\sum_{s \in \mathcal{F}(k)} d_{v,s}(k) x_{v,s}(k)$ and $\theta_L(k) = \sum_{s \in \mathcal{F}(k)\cup \{0\}} \theta(k,s)$ are the  total NFV delay on the nodes and the total communication delay on the links of flow $ k $, respectively.
%\vspace{0.1cm}\\
%{\bf\noindent$\bullet$ Problem Formulation\vspace{0.1cm}\\}
\subsubsection{MILP Formulation}
\indent The network slicing problem is to minimize a weighted sum of the total power consumption of the whole cloud network (equivalent to the number of activated cloud nodes \cite{Chen2020}) and the total delay of all services:
%Based on the above analysis, we obtain the following problem formulation:
\begin{align}
	%\begin{aligned}
	 \min_{\substack{\boldsymbol{x},\,\boldsymbol{y},\,\boldsymbol{r},\\\boldsymbol{z},\,\boldsymbol{\theta}}} & ~~ \sum_{v \in \mathcal{V}}y_v + \sigma \sum_{k \in \mathcal{K}} (\theta_L(k) + \theta_N(k)) \nonumber \\
	   {\text{s.t.~~}} & ~~   \text{\eqref{onlyonenode}--\eqref{delayconstraint}}, \nonumber                                                         \\
	  &                  ~~ x_{v,s}(k),~y_v\in\{0,1\},\,\forall~v\in\mathcal{{V}}, ~k\in\mathcal{K}, ~s\in \mathcal{F}(k), \nonumber       \\
	  &                  ~~ r(k,s,p),~r_{ij}(k, s, p )\geq 0,~z_{ij}(k, s, p )\in \{0,1\}, \nonumber                                       \\
	  &                  ~~~~ \forall~(i,j)\in \mathcal{L},~k\in \mathcal{K},~s\in \mathcal{F}(k)\cup \{0\},~p \in \mathcal{P}, \nonumber \\
	  &                 ~~ \theta(k,s)\geq 0,~\forall~k \in \mathcal{K},~s \in \mathcal{F}(k)\cup \{0\},
	\label{mip}
	\tag{{\rm{MINLP}}}
	%\end{aligned}
\end{align}
where $\sigma$ is a small positive value, meaning that the total power consumption term is much more important than the total delay term. 
	Indeed, as shown in \cite[Proposition 2]{Chen2020}, the total number of activated cloud nodes returned by solving problem \eqref{mip} with a sufficiently small positive value $\sigma$ is equal to the total number of activated cloud nodes returned by solving problem \eqref{mip} with $\sigma=0$.
	We remark that incorporating the total delay term into the objective function is also important in the following sense.
	First, the problem of minimizing the total number of activated cloud nodes often has multiple solutions and the total delay term can be regarded as a regularizer to make the problem have a unique solution as observed in our simulation results.
	Second, it will help to avoid cycles in each path connecting the two adjacent service functions in the solution of the problem.
	\revv{Third, adding the total delay term can help to return a solution with a much smaller delay, as illustrated in \cite{Chen2020}.}
	
%where $ \sigma $ is a constant value that balances the two terms in the objective function.
%

\revv{The above problem is a mixed integer nonlinear programming \eqref{mip} problem due to the nonlinearity in constraints \eqref{nonlinearcons} and \eqref{maxdelay1}.}
However, constraint \eqref{nonlinearcons} can be equivalently linearized \cite{Glover1975}.
Indeed, from \eqref{relalambdaandx11}, we have $ 0\leq r(k, s,p)  \leq 1 $, which, together with $ z_{ij}(k,s,p) \in \{0,1\} $, implies that
\begin{equation*}
	%\label{bilinear}
	0 \leq r_{ij}(k,s,p)  \leq r(k, s,p)  \leq 1.
\end{equation*}
Then constraint \eqref{nonlinearcons} can be equivalently reformulated as:
\begin{align}
	  & r_{ij}(k,s,p) \geq  z_{ij}(k, s,p) + r(k,s,p) -1,    \label{bilinear1}%\tag{\rm{5a}}
	\\
	  & r_{ij}(k,s,p) \leq   z_{ij}(k, s,p) ,                          \label{bilinear2}%\tag{\rm{5b}}
	\\
	  & r_{ij}(k,s,p) \leq  r(k,s,p).                                    \label{bilinear3}%\tag{\rm{5c}}
\end{align}
Together with $r_{ij}(k,s,p) \geq 0$, the above three constraints ensure that if $ z_{ij}(k, s,p) =1 $, we have $  r_{ij}(k,s,p) = r(k,s,p) $; otherwise $  r_{ij}(k,s,p) = 0  $. 
\revv{
	Constraint \eqref{maxdelay1} can also be equivalently linearized as 
	\begin{align}
		& \theta(k,s) \geq \sum_{(i,j) \in \mathcal{{L}}}  d_{ij}  z_{ij}(k, s, p),\nonumber                                                      \\
		& \qquad\qquad  \forall~k \in \mathcal{K}, ~ \forall~s \in \mathcal{F}(k) \cup \{0\},~\forall ~p \in \mathcal{P} \label{consdelay2funs1}.
	\end{align}
	Obviously,  constraint \eqref{maxdelay1} implies \eqref{consdelay2funs1}.
	On the other hand, due to the nonnegative objective coefficient $\sigma$ of $\theta(k,s)$ in the objective function, there always exists an optimal solution $(\boldsymbol{x}, \boldsymbol{y}, \boldsymbol{r}, \boldsymbol{z}, \boldsymbol{\theta})$ of problem \eqref{mip} (with \eqref{maxdelay1} replaced by \eqref{consdelay2funs1}) 
	such that  \eqref{maxdelay1} holds for all $k\in \mathcal{K}$ and $s\in \mathcal{F}(k)\cup \{0\}$.} %
\revv{Hence} we can present an equivalent mixed integer \emph{linear} programming formulation:
\begin{align}
 	\min_{\substack{\boldsymbol{x},\,\boldsymbol{y},\,\boldsymbol{r},\\\boldsymbol{z},\,\boldsymbol{\theta}}}   & ~~\sum_{v \in \mathcal{V}}y_v + \sigma \sum_{k \in \mathcal{K}} (\theta_L(k) + \theta_N(k)) \nonumber                                    \\
	   {\text{s.t.~~}}& ~~ \text{\eqref{onlyonenode}--\eqref{relalambdaandx11}}, \revv{~\text{\eqref{linkcapcons1}--\eqref{mediacons2}},~ \text{\eqref{delayconstraint}--\eqref{consdelay2funs1}}},\nonumber \\
	  &                  ~~ x_{v,s}(k),~y_v\in\{0,1\},\,\forall~v\in\mathcal{{V}}, ~k\in\mathcal{K}, ~s\in \mathcal{F}(k), \nonumber                                     \\
	  &                  ~~ r(k,s,p),~r_{ij}(k, s, p )\geq 0,~z_{ij}(k, s, p )\in \{0,1\}, \nonumber                                                                     \\
	  &                  ~~~~ \forall~(i,j)\in \mathcal{L},~k\in \mathcal{K},~s\in \mathcal{F}(k)\cup \{0\},~p \in \mathcal{P}, \nonumber                               \\
	  &                  ~~ \theta(k,s)\geq 0,~\forall~k \in \mathcal{K},~ s \in \mathcal{F}(k)\cup \{0\}. \tag{\rm{MILP}}
	\label{newmip}
\end{align}
Problem \eqref{newmip} can be solved to global optimality using standard MILP solvers like Gurobi \cite{Gurobi}.

It is worthwhile highlighting that  formulation \eqref{newmip} is equivalent to our recently proposed formulation (NS-II) in \cite{Chen2020}.
However, to present the SFC constraints (cf. \eqref{mediacons2}), reference \cite{Chen2020} introduced an equivalent but larger virtual network, and the formulation was built upon the virtual network.
In this paper, we do not introduce the virtual network, and the proposed formulation \eqref{newmip} is built upon the original network.
As a result, the number of variables and constraints in the proposed formulation \eqref{newmip}  are much smaller than that in \cite{Chen2020}.

%As observed in our experiment, solving formulation \eqref{newmip} is much faster than solving that in reference \cite{Chen2020}.

\subsection{Complexity Analysis}

The computational complexity of the (related) network slicing problem has been studied in the literature. 
In particular, references \cite{Amaldi2016} and \cite{Rost2020} studied a related problem where each slice (service) can be a general virtual graph %(a graph that is not limited to be a chain) 
and showed that the corresponding problem is strongly NP-hard when $P =1$ (i.e., only a single path is allowed to transmit the data flow of each service).
References \cite{Zhang2017} and \cite{Paschos2018} showed the strong NP-hardness of problem \eqref{newmip} when the nodes' capacities are limited and there is only a single service (but the number of functions in the services' SFC is not a constant), respectively.
The following theorem, however, shows the (strong) NP-hardness of problem \eqref{newmip} in two new very special cases and thus reveals the intrinsic difficulty of solving it.
This motivates us to develop efficient algorithms for approximately solving problem \eqref{newmip}, especially when the problem's dimension is large.

\begin{theorem}
	\label{NPhardness}
	(i) Problem \eqref{newmip} is NP-hard even when there is only a single service and there is only a single function in the SFC of this service.
	(ii) Problem \eqref{newmip} is strongly NP-hard even when each node's capacity, link's capacity, and service's E2E delay threshold are \rev{infinite}.
	%Moreover, there does not exist a constant-ratio approximation algorithm to solve it in this case.
\end{theorem}
\begin{proof}
	\rev{The proof can be found in Section I of \cite{Chen2021d}.}
\end{proof}

\section{Proposed Novel LP Relaxation}
\label{novellp}

In this section, we shall derive a novel LP relaxation for the network slicing problem \eqref{newmip}, which will be employed in the proposed \LPRR algorithm in Section \ref{sec:heur}.
\subsection{Natural LP Relaxation}
As the network slicing problem can be formulated as the MILP problem \eqref{newmip}, simply relaxing the binary variables $\{y_v\}$, $\{x_{v,s}(k)\}$, and $\{z_{ij}(k,s,p)\}$ to continuous variables in $[0,1]$ will give a natural LP relaxation:
\begin{align}
 	\min_{\substack{\boldsymbol{x},\,\boldsymbol{y},\,\boldsymbol{r},\\\boldsymbol{z},\,\boldsymbol{\theta}}} &  ~~ \sum_{v \in \mathcal{V}}y_v + \sigma \sum_{k \in \mathcal{K}} (\theta_L(k) + \theta_N(k)) \nonumber                                    \\
	   {\text{s.t.~~}} & ~~ \text{\eqref{onlyonenode}--\eqref{relalambdaandx11}}, ~ \revvv{\text{\eqref{linkcapcons1}--\eqref{mediacons2}},~ \text{\eqref{delayconstraint}--\eqref{consdelay2funs1}}}, \nonumber \\
	  &                  ~~ x_{v,s}(k),~y_v\in[0,1],\,\forall~v\in\mathcal{{V}}, ~k\in\mathcal{K}, ~s\in \mathcal{F}(k), \nonumber                                        \\
	  &                  ~~ r(k,s,p),~r_{ij}(k, s, p )\geq 0,~z_{ij}(k, s, p )\in [0,1], \nonumber                                                                        \\
	  &                  ~~ ~~ \forall~(i,j)\in \mathcal{L},~k\in \mathcal{K},~s\in \mathcal{F}(k)\cup \{0\},~p \in \mathcal{P}, \nonumber                                \\
	  &                  ~~ \theta(k,s)\geq 0,~\forall~k \in \mathcal{K}, ~s \in \mathcal{F}(k)\cup \{0\}.\tag{\rm{LP-I}}
	\label{milplp}
\end{align}
However, the above natural LP relaxation has two disadvantages, which are detailed as follows.

First, \emph{the natural LP relaxation \eqref{milplp} is very weak in terms of providing a poor LP bound}.
Indeed, while it has been shown that the three constraints \eqref{bilinear1}--\eqref{bilinear3} and $r_{ij}(k,s,p) \geq 0$ are equivalent to the nonlinear constraint \eqref{nonlinearcons} when $z_{ij}(k,s,p) \in \{0,1\}$,  it can only be shown that \eqref{bilinear1}--\eqref{bilinear3} is a relaxation of the nonlinear constraint \eqref{nonlinearcons} when $z_{ij}(k,s,p) \in (0,1)$.
Therefore, the feasible region of the natural relaxation \eqref{milplp} can be larger than that of the nonlinear relaxation of problem \eqref{mip}:
\begin{align}
 \min_{\substack{\boldsymbol{x},\,\boldsymbol{y},\,\boldsymbol{r},\\\boldsymbol{z},\,\boldsymbol{\theta}}} &  ~~ \sum_{v \in \mathcal{V}}y_v + \sigma \sum_{k \in \mathcal{K}} (\theta_L(k) + \theta_N(k)) \nonumber \\
	   {\text{s.t.~~}} &   ~~ \text{\eqref{onlyonenode}--\eqref{delayconstraint}}, \nonumber                                                         \\
	  &                  ~~ x_{v,s}(k),~y_v\in[0,1],\,\forall~v\in\mathcal{{V}},~k\in\mathcal{K}, ~s\in \mathcal{F}(k), \nonumber          \\
	  &                  ~~ r(k,s,p),~r_{ij}(k, s, p )\geq 0,~z_{ij}(k, s, p )\in [0,1], \nonumber                                         \\
	  &                 ~~ ~~ \forall~(i,j)\in \mathcal{L},~k\in \mathcal{K},~s\in \mathcal{F}(k)\cup \{0\},~p \in \mathcal{P}, \nonumber \\
	  &                  ~~ \theta(k,s)\geq 0,~\forall~k \in \mathcal{K},~s \in \mathcal{F}(k)\cup \{0\},
	\label{nlp}
	\tag{\rm{NLP}}
\end{align}
and hence the natural LP relaxation \eqref{milplp} is weaker than the nonlinear relaxation \eqref{nlp}.
In other words, the optimal objective value of relaxation \eqref{milplp} can be smaller than that of relaxation \eqref{nlp}.
Please see the following example for an illustration.
\begin{example}
	\label{example1}
	Consider the toy example in Fig. \ref{example}.
	There are two links from node $S$ to node $D$.
	The communication capacities are both $0.5$ and the communication delays are $1$ and $2$, respectively.
	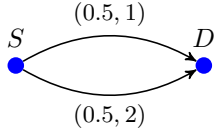
\begin{figure}[t]
		\begin{center}
			\begin{tikzpicture}[->,>=stealth',shorten >=0pt,auto,node distance=10cm,
					semithick]
				\tikzstyle{every state}=[fill=blue,draw=none,scale=0.25,text=white]
				\node[state,label=above:{$S$}] (A)                  {};
				\node[state,label=above:{$D$}]         (B) [right of=A] {};

				\path
				(A) edge[bend left]              node {\small$(0.5,1)$} (B)
				(A) edge[bend right,below]              node {\small$(0.5,2)$} (B)
				;
			\end{tikzpicture}
			\caption{A toy network example where the pair $(c,d)$ over each link denotes that the communication capacity and the communication delay of the corresponding link are $c$ and $d$, respectively.}
			\label{example}
		\end{center}
	\end{figure}
	Suppose that there is a flow with its source and destination nodes being  $S$ and $D$, respectively, and its data rate being $1$.
	For simplicity of presentation, we assume that there does not exist any function in the flow's SFC and the problem is to route the flow to minimize the delay from node $S$ to node $D$.
	\revv{
		Since the capacities on both of links in Fig. \ref{example} are equal to 0.5, the traffic flow from the node $S$ to the node $D$ with data rate being 1 must use two paths (i.e., one uses the top link and the other uses the bottom link) on which both data rates are 0.5.
		As presented in \eqref{maxdelay1}, the E2E delay of a traffic flow between two nodes is the largest delay of all paths used by this flow. 
		Consequently,  the E2E delay of this traffic flow is equal to $\max\{1,2 \}=2$ (as the delays of the top and bottom links are equal to $1$ and $2$, respectively).
	}
	
	Let $P=2$ in relaxations \eqref{milplp} and \eqref{nlp}.
	Solving the nonlinear relaxation \eqref{nlp} returns a solution with the E2E delay (i.e., the objective value) being $1.5$, whereas
	solving the natural LP relaxation \eqref{milplp} returns a solution with the E2E delay being $1.25$; see Appendix \ref{detailsofex1} for more details.
	This example clearly shows that the natural LP relaxation \eqref{milplp} can be strictly weaker than the nonlinear relaxation \eqref{nlp}.
\end{example}

\revv{We remark that the weakness of problem \eqref{milplp} is due to the split of traffic flows in problem \eqref{newmip}, i.e., multiple paths are possibly used to carry out the traffic flow between the two nodes hosting two adjacent functions.
Indeed, if only a single path is used to carry out the traffic flow between any two nodes hosting two adjacent functions (i.e., $\mathcal{P}=\{1\}$ in \eqref{newmip}), then the two relaxations \eqref{nlp} and \eqref{milplp} \revvv{are} equivalent. 
In this case, \eqref{relalambdaandx11} and \eqref{nonlinearcons} reduce to 
$r(k,s,1) =1$ and $r_{ij}(k,s,1)= z_{ij}(k,s,1)$, respectively.
This, together with $z_{ij}(k,s,1) \in [0,1]$, implies \eqref{bilinear1}--\eqref{bilinear3} (with $p=1$).
Conversely, by \eqref{bilinear1}, \eqref{bilinear2}, and $r(k,s,1) =1$, we have $r_{ij}(k,s,1)= z_{ij}(k,s,1)$, showing that \eqref{nonlinearcons} holds.
}

Second, \emph{the dimension of the natural LP relaxation \eqref{milplp} is very large and there is some redundancy in it}.
To see this, recall that in problem \eqref{newmip}, in order to model different paths for flow $(k,s)$, we introduce the notation $\{ p  \mid  p \in \mathcal{P} \}$ and use $\{(i,j) \mid r_{ij}(k,s,p)>0\}$ to represent the $p$-th path of flow $(k,s)$ (cf. \eqref{nonlinearcons} and \eqref{mediacons2}).
However, as $z_{ij}(k,s,p) \in [0,1]$ in the natural LP relaxation \eqref{milplp}, the traffic flow \revv{$\{(i,j) \mid r_{ij}(k,s,p)>0\}$} can be split into multiple paths.
This reveals that there is some redundancy in the natural LP relaxation \eqref{milplp}, i.e., we do not need to introduce the notation $\{p \mid p\in \mathcal{P}\}$ to model different paths for flow $(k,s)$.

The weakness and the large problem size of relaxation \eqref{milplp} make it unsuitable to be embedded in an LP rounding algorithm, i.e., it can lead to a bad performance of an LP rounding algorithm; see the result in Section \ref{sec_num} further ahead.
\subsection{Novel LP Relaxation}
To overcome the above two disadvantages of relaxation \eqref{milplp}, we propose a new LP relaxation for problem \eqref{newmip} in this subsection.
\revv{We begin with the following lemma stating a property of relaxation \eqref{nlp}, which plays a crucial role in deriving the new LP relaxation.}
\begin{lemma}
	\label{lemma1}
	\revv{Let $(\boldsymbol{x},\boldsymbol{y},\boldsymbol{r},\boldsymbol{z},\boldsymbol{\theta})$ be a feasible solution of  relaxation \eqref{nlp}.
	Then 
	\begin{align}
		& \bar{r}_{ij}(k,s) := \sum_{p \in \mathcal{P}}r_{ij}(k,s,p)\leq 1,\nonumber\\
		& \qquad~~~ \forall~(i,j)\in \mathcal{L},~\forall~k \in \mathcal{K}, ~\forall~s \in \mathcal{F}(k)\cup \{0\}, \label{rbardef}
	\end{align}
	 and $(\boldsymbol{x}',\boldsymbol{y}',\boldsymbol{r}',\boldsymbol{z}',\boldsymbol{\theta}')$ defined by 
	\begin{align}
		& \boldsymbol{x}'=\boldsymbol{x},~ \boldsymbol{y}'=\boldsymbol{y},\label{xydef}\\
		& r'(k,s, p)~=\left\{\begin{array}{ll}1,
			& \qquad\qquad {\text{if}}~ p=1 ;~ \\
			0, & \qquad\qquad\text{otherwise},\end{array}\right.\label{rdef11}                                                                          \\
		& r'_{ij}(k,s, p)=\left\{\begin{array}{ll}\bar{r}_{ij}(k,s),
			& \quad {\text{if}}~ p=1 ;~ \\
			0, & \quad {\text{otherwise}},\end{array}\right.	\label{rdef21}                                                                 \\
		& z'_{ij}(k,s, p)=\bar{r}_{ij}(k,s),\label{zdef1}                                                                                             \\
		& \theta'(k,s) = \sum_{(i,j) \in \mathcal{L}} d_{ij}\bar{r}_{ij}(k,s), \label{thetadef1}\\
		& ~~~ \forall~(i,j)\in \mathcal{L},~\forall~k \in \mathcal{K}, ~\forall~s \in \mathcal{F}(k)\cup \{0\}, ~\forall~p \in \mathcal{P},\nonumber
	\end{align}
	is also a feasible solution of relaxation \eqref{nlp} with at most the same objective value as that of $(\boldsymbol{x},\boldsymbol{y},\boldsymbol{r},\boldsymbol{z},\boldsymbol{\theta})$.
}
\end{lemma}
\begin{proof}
	\revv{
	For each $(i,j) \in \mathcal{L}$, $k\in \mathcal{K}$, and $s \in \mathcal{F}(k)$, we have $\bar{r}_{ij}(k,s) \leq 1$ as 
	\begin{align*}
		%		\bar{z}_{ij}(k,s) & ~ =   \sum_{p \in \mathcal{P}} r_{ij}(k,s,p) \\
		\sum_{p \in \mathcal{P}} r_{ij}(k,s,p) & ~=\sum_{p \in \mathcal{P}} r(k,s,p) z_{ij}(k,s,p) ~( \text{from}~\eqref{nonlinearcons})\\
		& ~\leq\sum_{p \in \mathcal{P}} r(k,s,p) ~( \text{from}~z_{ij}(k,s,p) \leq 1)                \\
		& ~= 1.~( \text{from}~\eqref{relalambdaandx11})
	\end{align*}
	By the definition of $(\boldsymbol{x}',\boldsymbol{y}',\boldsymbol{r}',\boldsymbol{z}',\boldsymbol{\theta}')$ and the feasibility of $(\boldsymbol{x},\boldsymbol{y},\boldsymbol{r},\boldsymbol{z},\boldsymbol{\theta})$ for relaxation (NLP), 
	constraints \eqref{onlyonenode}--\eqref{linkcapcons1} and \eqref{maxdelay1} hold at $(\boldsymbol{x}',\boldsymbol{y}',\boldsymbol{r}',\boldsymbol{z}',\boldsymbol{\theta}')$.}
	%and (ii) $(\boldsymbol{x},\boldsymbol{y},\boldsymbol{r},\boldsymbol{z},\boldsymbol{\theta})$  and  $(\boldsymbol{x}',\boldsymbol{y}',\boldsymbol{r}',\boldsymbol{z}',\boldsymbol{\theta}')$ yield the same objective value.
%	It remains to show that constraints \eqref{mediacons2} and \eqref{maxdelay1} also hold at $(\boldsymbol{x}',\boldsymbol{y}',\boldsymbol{r}',\boldsymbol{z}',\boldsymbol{\theta}')$. 
	\revv{Constraint \eqref{mediacons2} holds true since
\begin{equation*}
		\begin{aligned}
			& & & \sum_{j: (j,i) \in \mathcal{{L}}} z'_{ji}(k, s,p) - \sum_{j: (i,j) \in \mathcal{{L}}} z'_{ij}(k, s,p)\\
			& = &   & \sum_{j: (j,i) \in \mathcal{{L}}} \sum_{p \in \mathcal{P}}r_{ji}(k, s,p) - \sum_{j: (i,j) \in \mathcal{{L}}} \sum_{p \in \mathcal{P}}r_{ij}(k, s,p)  \\
			&   &   & \qquad \qquad\qquad\qquad\qquad\qquad\qquad ( \text{from}~\eqref{rbardef}~\text{and}~\eqref{zdef1})                                                                 \\
			& = &   & \sum_{j: (j,i) \in \mathcal{{L}}} \sum_{p \in \mathcal{P}}r(k,s,p)z_{ji}(k, s,p) -                                                                   \\
			&   &   & \qquad\qquad \sum_{j: (i,j) \in \mathcal{{L}}} \sum_{p \in \mathcal{P}}r(k,s,p)z_{ij}(k, s,p)~( \text{from}~\eqref{nonlinearcons})                   \\
			& = &   & \sum_{p \in \mathcal{P}}r(k,s,p) \left ( \sum_{j: (j,i) \in \mathcal{{L}}} z_{ji}(k, s,p) - \sum_{j: (i,j) \in \mathcal{{L}}}z_{ij}(k, s,p) \right ) \\
			& = &   & \sum_{p \in \mathcal{P}}r(k,s,p) \left (  \left\{\begin{array}{ll}0,
				& {\text{if}}~ i \in \mathcal{{I}}\backslash \mathcal{{V}} ;~ \\
				x_{i,s+1}(k) - x_{i,s}(k), & {\text{if}}~i \in \mathcal{{V}},\end{array}\right.\right )                                                           \\
			&   &   & \qquad \qquad\qquad\qquad\qquad\qquad\qquad\qquad\qquad ( \text{from}~\eqref{mediacons2})                                                            \\
			& = &   & \left\{\begin{array}{ll}0,
				& {\text{if}}~ i \in \mathcal{{I}}\backslash \mathcal{{V}} ;~ \\
				x_{i,s+1}(k) - x_{i,s}(k), & {\text{if}}~i \in \mathcal{V}.\end{array}\right. ~( \text{from}~\eqref{relalambdaandx11})
		\end{aligned}
	\end{equation*}
To prove that constraint \eqref{delayconstraint} holds true at $(\boldsymbol{x}',\boldsymbol{y}',\boldsymbol{r}',\boldsymbol{z}',\boldsymbol{\theta}')$ and the objective value of $(\boldsymbol{x}',\boldsymbol{y}',\boldsymbol{r}',\boldsymbol{z}',\boldsymbol{\theta}')$ is less than or equal to that of
	$(\boldsymbol{x},\boldsymbol{y},\boldsymbol{r},\boldsymbol{z},\boldsymbol{\theta})$,}
	\revv{it suffices to show
	\begin{equation*}
		\begin{aligned}
			%& & & \sum_{(i,j) \in \mathcal{{L}}}  d_{ij}  z'_{ij}(k, s,p)\\
			\theta'(k,s)& =    &   & \sum_{(i,j) \in \mathcal{{L}}}  d_{ij} \sum_{p \in \mathcal{P}}  r_{ij}(k, s,p) ~( \text{from}~\eqref{rbardef}~\text{and}~\eqref{thetadef1})                  \\
			& =    &   & \sum_{(i,j) \in \mathcal{{L}}}  d_{ij} \sum_{p \in \mathcal{P}}  r(k,s,p) z_{ij}(k, s,p) ~( \text{from}~\eqref{nonlinearcons}) \\
			& =    &   & \sum_{p \in \mathcal{P}}  r(k,s,p)  \sum_{(i,j) \in \mathcal{{L}}}  d_{ij} z_{ij}(k, s,p)                                      \\
			& \leq &   & \sum_{p \in \mathcal{P}}  r(k,s,p) \theta(k,s) ~( \text{from}~\eqref{maxdelay1})                                         \\
			& = &   & \theta(k,s).~( \text{from}~\eqref{relalambdaandx11})                                                                  %         \\
			%& =    &   & \theta'(k,s).
		\end{aligned}
	\end{equation*}
The proof is complete.}
\end{proof}
\revv{The following corollary follows directly from Lemma \ref{lemma1}.}
\begin{corollary}\label{corollary}
	\revv{Suppose that relaxation \eqref{nlp} has a feasible solution.}
	\revv{Then there exists an optimal solution $(\boldsymbol{x},\boldsymbol{y},\boldsymbol{r},\boldsymbol{z},\boldsymbol{\theta})$ of relaxation \eqref{nlp} such that}
	\begin{align}
		& \revv{r(k,s,1) = 1, ~\forall ~k\in\mathcal{K},~s \in \mathcal{F}(k)\cup \{0\},\label{tmpeq1}}\\
		& \revv{r_{ij}(k,s,p)=r(k,s,p)=0,\nonumber} \\
		& \revv{\quad \forall ~k\in\mathcal{K},~s \in \mathcal{F}(k)\cup \{0\}, ~p \in \mathcal{P}\backslash\{1\},~(i,j) \in \mathcal{L},\label{tmpeq2}}\\
		& \revv{z_{ij}(k,s,p)=r_{ij}(k,s,1),\nonumber}\\
		& \revv{\qquad\quad \forall ~k\in\mathcal{K},~s \in \mathcal{F}(k)\cup \{0\}, ~p \in \mathcal{P},~(i,j) \in \mathcal{L}.}\label{tmpeq3}
	\end{align}
\end{corollary} 
\revv{By Corollary \ref{corollary}, we can remove variables $\{r(k,s,p)\}$, $\{z_{ij}(k,s,p)\}$, and $\{ r_{ij}(k,s,p) ~\text{with}~p\in \mathcal{P}\backslash\{1\}\}$ (using Eqs. \eqref{tmpeq1}--\eqref{tmpeq3}) from problem \eqref{nlp}.}
%The proposed LP relaxation is based on the observation that for each $p\in \mathcal{P}$, the traffic flow  $\{(i,j) :r_{ij}(k,s,p)>0\}$  can be split into multiple paths when binary variables $\{z_{ij}(k,s,p)\}$ are relaxed to continuous variables in $[0,1]$.
%
%\revv{Lemma \ref{lemma1} indicates that there exists an optimal solution $(\boldsymbol{x}',\boldsymbol{y}',\boldsymbol{r}',\boldsymbol{z}',\boldsymbol{\theta}')$ of relaxation \eqref{nlp} such that $r'(k,s,1)=1$ and $ r'(k,s,p)=0$ for all $k \in \mathcal{K}$, $s \in \mathcal{F}(k)$, and $p \in \mathcal{P}\backslash \{1\}$.
%This indicates that for such a solution of relaxation \eqref{nlp}, only $(k,s,1)$ carries traffic flow of $(k,s)$ with positive data rates (the other $(k,s,p)$, $p \in \mathcal{P}\backslash \{1\}$, do not carry traffic flow of $(k,s)$ with positive data rates).
%This property of relaxation \eqref{nlp} enables to derive a new LP relaxation for problem \eqref{newmip}.
%Specifically, we set $P=1$ and $\mathcal{P}=\{1\}$ in the nonlinear relaxation \eqref{nlp}.}
%\todo[inline]{WC: Rewrite this part after the discussion.}
%
%Then by \eqref{relalambdaandx11}, we have $r(k,s,1) = 1$, and hence constraint \eqref{nonlinearcons} reduces to $r_{ij}(k,s,1)=z_{ij}(k,s,1)$.
In addition, constraints \eqref{relalambdaandx11} and \eqref{nonlinearcons} can be removed and  constraints \eqref{linkcapcons1}--\eqref{maxdelay1} can be replaced by
\begin{equation}
	\label{linkcapcons2}
	\tag{\rm{6}'}
	\sum_{k \in \mathcal{K}} \sum_{s\in \mathcal{F}(k) \cup \{0\}} \lambda_{s}(k) r_{ij}(k, s,1) \leq C_{ij}, ~  \forall~(i,j) \in \mathcal{L},
\end{equation}
\eqref{mediacons2-2} in the next page, 
\begin{figure*}[t]
	\begin{equation}
		\tag{\rm{7}'}
		\sum_{j: (j,i) \in \mathcal{{L}}} r_{ji}(k, s, 1) - \sum_{j: (i,j) \in \mathcal{{L}}} r_{ij}(k, s, 1)=\left\{\begin{array}{ll}
			0,                        & \text{if}~ i   \in {\mathcal{I}}\backslash {\mathcal{V}}; \\
			x_{i,s+1}(k)- x_{i,s}(k), & \text{if}~  i \in {\mathcal{V}},
		\end{array} \right.	\forall~k \in \mathcal{K},~\forall~s \in \mathcal{F}(k)\cup \{0\} \label{mediacons2-2}.
	\end{equation}
	\vspace{-5mm}
	\hrulefill
	%\vspace{-5mm}
	%\vspace*{-4pt}
\end{figure*}
and
\begin{align}
	  & \theta(k,s)~\revv{=}~\sum_{(i,j) \in \mathcal{{L}}}  d_{ij}  r_{ij}(k, s, 1),\nonumber                                                                   \\
	  & \qquad\qquad \qquad\qquad \forall~k \in \mathcal{K}, ~ \forall~s \in \mathcal{F}(k) \cup \{0\} \label{consdelay2funs2}.\tag{\rm{8}'}
\end{align}
Then problem \eqref{nlp} reduces to the following LP problem:
\begin{align}
	%\begin{aligned}
	   \min_{\boldsymbol{x},\,\boldsymbol{y},\,\boldsymbol{r},\,\boldsymbol{\theta}} & ~~   \sum_{v \in \mathcal{V}}y_v + \sigma \sum_{k \in \mathcal{K}} (\theta_L(k) + \theta_N(k)) \nonumber                       \\
	   {\text{s.t.~~}}                                                             &~~    \text{\eqref{onlyonenode}--\eqref{nodecapcons}},~ \text{\eqref{linkcapcons2}--\eqref{consdelay2funs2}},~ \eqref{delayconstraint}, \nonumber \\
	  &                                                                                ~~x_{v,s}(k),~y_v\in[0,1],\,\forall~v\in \mathcal{V}, ~k\in\mathcal{K},~s\in \mathcal{F}(k), \nonumber                      \\
	  & ~~                                                                               r_{ij}(k, s,1 )\in [0,1], ~\theta(k,s)\geq 0,\nonumber                                                                    \\
	  &~~                                                                                ~~~ \forall~(i,j)\in \mathcal{{L}}, ~k\in \mathcal{K},~s\in \mathcal{F}(k)\cup \{0\}.
	\label{lp}
	\tag{\rm{LP-II}}
	%\end{aligned}
\end{align}
\revv{Based on the above derivation, we immediately have the following theorem.}
\begin{theorem}
	\label{LPrelaxation}
	The LP problem \eqref{lp} is equivalent to the nonlinear relaxation \eqref{nlp} with $P\geq 1$.
\end{theorem}
\begin{example}
	\revv{Consider the toy example in Fig. \ref{example}.
		Let $P=2$ in relaxation \eqref{lp}.
		Solving relaxation \eqref{lp} returns a solution with the E2E delay (i.e., the objective value) being $1.5$, which is the same as the optimal value of relaxation \eqref{nlp}; see Appendix A for more details.
		This example shows that relaxations \eqref{nlp} and \eqref{lp} provide the same lower bound for formulation \eqref{newmip}.}
\end{example}

\rev{
	Next, we provide more insight into the derivation of LP relaxation \eqref{lp}. In fact, there are two sources of the nonconvexity of problem \eqref{mip}: the binary placement constraints and the multi-path routing and E2E delay constraints.
	\eqref{lp} can be derived by using convex techniques to deal with these two nonconvex constraints.  
	%Problem \eqref{lp} can be developed by using standard convex techniques to deal with these two nonconvexity constaints
	%we deal with these two nonconvexity constaints one by one.
	The way to deal with the first nonconvex constraint is to relax the binary variables $\boldsymbol{x}$ and $\boldsymbol{y}$ into continuous variables in $[0,1]$.
	As for the second one, one can remove the multi-path routing constraint (which mathematically corresponds to using the continuous variables $\{r_{ij}(k,s,1)\}$ and the related constraints \eqref{linkcapcons2}--\eqref{mediacons2-2}) and replace the largest delay of the $P$ paths in flow $(k,s)$ with the average delay (which mathematically corresponds to using constraint \eqref{consdelay2funs2}).
}

The equivalence of problems \eqref{lp} and \eqref{nlp} in Theorem \ref{LPrelaxation} (i) implies that the LP problem \eqref{lp} can also be seen as a relaxation of problem \eqref{newmip} and \rev{(ii) characterizes how tight relaxation \eqref{lp} can possibly be, i.e., relaxation \eqref{lp} provides a lower bound of the network slicing problem that is as tight as that relaxation \eqref{nlp} can provide.}
However, in sharp contrast to relaxation \eqref{nlp} which is a nonconvex problem (due to the bilinearity in constraint \eqref{nonlinearcons}), relaxation \eqref{lp} is a (convex) LP problem and can be (globally and efficiently) solved using the polynomial time interior-point method in  \cite{Ben-Tal2001}.
In addition, when compared with the natural LP relaxation \eqref{milplp}, the novel LP relaxation \eqref{lp} enjoys the following two advantages.
First, LP relaxation \eqref{lp} is also stronger than LP relaxation \eqref{milplp}, which follows from the fact that nonlinear relaxation \eqref{nlp} is stronger than  LP relaxation \eqref{milplp} and the equivalence of \eqref{nlp} and \eqref{lp} in Theorem \ref{LPrelaxation}.
\begin{corollary}
	\label{strongerLP}
	The novel LP relaxation \eqref{lp} is stronger than the natural LP relaxation \eqref{milplp}.
\end{corollary}
\noindent Second, LP relaxation \eqref{lp} is much more compact than LP relaxation \eqref{milplp}.
Indeed, the numbers of variables and constraints in  LP relaxation \eqref{lp} are $\mathcal{O} (|\mathcal{L}|\sum_{k \in \mathcal{K}} \ell_k)$ and $\mathcal{O}(\min\{|\mathcal{I}|\sum_{k\in \mathcal{K}}\ell_k, |\mathcal{L}|\})$, respectively.
These are much smaller than those in LP relaxation \eqref{milplp}, which are $\mathcal{O} (|\mathcal{L}||\mathcal{P}|\sum_{k \in \mathcal{K}} \ell_k)$.
%As it will be seen in Section \ref{sec_num}, 
The strongness and compactness of the proposed LP relaxation \eqref{lp} play a crucial role in the effectiveness and efficiency of the proposed \LPRR algorithm in the next section.

\subsection{Integrality Gap of Relaxation \eqref{lp}}
\label{integralitygapanalysis}

In this subsection, we analyze the integrality gap of relaxation \eqref{lp} in some special cases. 
In particular, we show that the integrality gap of relaxation \eqref{lp} in some special case (i.e., case (i) in the following Theorem \ref{integralitygap}) can be infinite, which provides more insight into the intrinsic difficulty of the original problem.
%reveals the intrinsic difficulty of the original problem.
In addition, we also identify \revv{three} other special cases (i.e., \revv{cases (ii)--(iv)} in the following Theorem \ref{integralitygap}) where we can derive bounded or zero integrality gaps.
%These analysis results provide more insight into the tractability/in of problem \eqref{newmip}.

Let us first specify the special cases of interest in this subsection: each link's capacity is infinite, each link's delay is zero, and the SFC of each service contains only one function in problem \eqref{newmip}. 
In this case, the capacity constraint on all links and the E2E delay constraint of all services are automatically satisfied.
Let $\mathcal{V}(k)\subseteq \mathcal{V}$ denote the set of cloud nodes that can process the function $f_1^k$, which is the (only) function in the SFC of service $k$.
Without loss of generality, we assume that there exists a path between the source node $S(k)$ and each node in $\mathcal{V}(k)$ and between each node in $\mathcal{V}(k)$ and the destination node $D(k)$ and 
the processing capacity of node $v$ is larger than or equal to the data rate $\lambda_1(k)$ of function $f_1^k$, i.e., 
$\mu_v \geq \lambda_1(k)$ (otherwise, we can set $x_{v,1}(k)=0$).
Then, we must have 
\begin{equation}
\label{noprocesscons}
	%\tag{$\star$}
x_{v,1}(k) = 0,~\forall~k \in \mathcal{K},~ v \in \mathcal{V}\backslash \mathcal{V}(k).
\end{equation}
Hence problem \eqref{newmip} reduces to 
\begin{align}
%\begin{CF}
\min_{\substack{\boldsymbol{x},\,\boldsymbol{y}}} & ~~ \sum_{v \in \mathcal{V}}y_v \nonumber \\
~~{\text{s.t.~}} &~~   \text{\eqref{onlyonenode}--\eqref{nodecapcons}},~\eqref{noprocesscons},\nonumber\\ &~~\revv{y_v\in\{0,1\},\,\forall~v\in\mathcal{{V}}}\revvv{,}\nonumber                                                         \\
&~~    x_{v,1}(k)\in \{0,1\},~\forall~v\in\mathcal{{V}}, ~k\in\mathcal{K}, 
\label{SC}
%\tag{FC}
\end{align}
and relaxation \eqref{lp} reduces to 
\begin{align}
%\begin{aligned}
\min_{\substack{\boldsymbol{x},\,\boldsymbol{y}}} & ~~  \sum_{v \in \mathcal{V}}y_v \nonumber \\
~~{\text{s.t.~}} &~~    \text{\eqref{onlyonenode}--\eqref{nodecapcons}},~\eqref{noprocesscons},\nonumber\\
&~~\revv{y_v\in[0,1],\,\forall~v\in\mathcal{{V}}}, \nonumber                                                         \\
&~~    x_{v,1}(k)\in [0,1],~\forall~v\in\mathcal{{V}}, ~k\in\mathcal{K}.
\label{LP}
%\tag{\rm{FC-LP}}
%\end{aligned}
\end{align}

\begin{theorem}
	\label{integralitygap}
	Suppose that each link's capacity is infinite, each link's delay is zero, and the SFC of each service contains only one function in problem \eqref{newmip}. Then,
	\begin{itemize}
		\item [(i)] if $\mu_v = \infty$ for all $v \in \mathcal{V}$, the integrality gap of relaxation \eqref{lp} can be infinite;
		\item [(ii)] \revv{if $\mu_v = \infty$ for all $v \in \mathcal{V}$ and $|\mathcal{V}(k)| \leq u$ for all $k \in \mathcal{K}$, the integrality gap of relaxation \eqref{lp} is at most $u$};
		\item [(iii)] if $\lambda_1(k) \geq \frac{1}{t}\mu_v $ for all  $k \in \mathcal{K}$ and $v\in\mathcal{V}(k)$ where $t$ is some constant, the integrality gap of relaxation \eqref{lp} is at most $t$; and
		\item [(iv)]  if each node can host (at most) one function, the integrality gap of relaxation \eqref{lp} is zero and the original problem \eqref{newmip} can be solved to global optimality.
		
	\end{itemize}
\end{theorem}
\begin{proof}
		We first prove case (i).
		In this case, constraint \eqref{nodecapcons} can be removed.
%		\begin{equation}
%		\sum_{v \in \mathcal{V}(k)} y_v \geq 1,~\forall~k \in \mathcal{K}.
%		\end{equation}
		Then problems \eqref{SC} and \eqref{LP} reduce to
		\revv{
		\begin{align}
			\min_{\substack{\boldsymbol{x},\,\boldsymbol{y}}} & ~~ \sum_{v \in \mathcal{V}}y_v \nonumber \\
			~{\text{s.t.}} &~~  \eqref{onlyonenode},~\eqref{xyrelation},~\eqref{noprocesscons},\nonumber\\
			&~~y_v\in\{0,1\},\,\forall~v\in\mathcal{{V}}, \nonumber                                                         \\
			&~~    x_{v,1}(k)\in\{0,1\},~\forall~v\in\mathcal{{V}}, ~k\in\mathcal{K},
			\label{SC2}%\tag{3'}
		\end{align}
		and
		\begin{align}
			\min_{\substack{\boldsymbol{x},\,\boldsymbol{y}}} & ~~ \sum_{v \in \mathcal{V}}y_v \nonumber \\
			~{\text{s.t.}} &~~   \eqref{onlyonenode},~\eqref{xyrelation},~\eqref{noprocesscons},\nonumber\\
			&~~y_v\in[0,1],\,\forall~v\in\mathcal{{V}}, \nonumber                                                         \\
			&~~    x_{v,1}(k)\in[0,1],~\forall~v\in\mathcal{{V}}, ~k\in\mathcal{K}.
			\label{SC2-LP}%\tag{4'}
		\end{align}
		%Constraints \eqref{onlyonenode} and \eqref{xyrelation} can be replaced by $\sum_{v \in \mathcal{V}(k)} y_v \geq 1,~\forall~k \in \mathcal{K}.$
		Problems \eqref{SC2} and \eqref{SC2-LP} are further equivalent to}
			\begin{align}
			%\begin{aligned}
			\min_{\substack{\boldsymbol{y}}} &  ~~ \sum_{v \in \mathcal{V}}y_v \nonumber \\
			~~{\text{s.t.~}} &~~    \sum_{v \in \mathcal{V}(k)} y_v \geq 1,~\forall~k \in \mathcal{K}, \nonumber                                                         \\
			&~~    y_v\in\{0,1\},\,\forall~v\in\mathcal{{V}}, 
			\label{SC1}
			%\tag{SC}
			%\end{aligned}
		\end{align}
		and
		\begin{align}
			%\begin{aligned}
			\min_{\substack{\boldsymbol{y}}} & ~~  \sum_{v \in \mathcal{V}}y_v \nonumber \\
			~~{\text{s.t.~}} &~~    \sum_{v \in \mathcal{V}(k)} y_v \geq 1,~\forall~k \in \mathcal{K},  \nonumber                                                         \\
			& ~~    y_v\in[0,1],\,\forall~v\in\mathcal{{V}},
			\label{LP1}
			%\tag{SC-LP}
			%\end{aligned}
		\end{align}
		\revv{respectively, in the sense that the solution of one problem can be constructed from that of the other. 
		Indeed, if $(\boldsymbol{x}^*, \boldsymbol{y}^*)$ is a solution of problem \eqref{SC2-LP}, then $\boldsymbol{y}^*$ must be a solution of problem \eqref{LP1}. 
		Conversely, letting $\boldsymbol{y}^*$ be a solution of problem \eqref{LP1}, we can also construct a feasible solution $(\boldsymbol{x}^*, \boldsymbol{y}^*)$ of problem \eqref{SC2-LP} as follows. 
		Let  $\sigma_1(k), \ldots, \sigma_{|\mathcal{V}(k)|}(k)$ be a permutation of $ \mathcal{V}(k)$ such that $y^*_{\sigma_1(k)} \geq \cdots \geq  y^*_{\sigma_{|\mathcal{V}(k)|}(k)}$, and $t(k)$ be the largest integer such that
		$ \sum_{g=1}^{t(k)} y_{\sigma_g(k)}^* \leq 1$.
		By $\sum_{v \in \mathcal{V}(k)}y^*_v \geq 1$, we must have $t(k) \in \{1, \ldots, |\mathcal{V}(k)|\}$.
		For each $k \in \mathcal{K}$, if $t(k) = |\mathcal{V}(k)|$,
		define
		\begin{equation}
			\label{LPsolutions}
			\begin{aligned}
				x^*_{v,1}(k)= & \left\{\begin{array}{ll}y^*_v, &
					~ {\text{if}}~v = \sigma_{1}(k), \ldots, \sigma_{t(k)}(k);\\
					0, & ~{\text{otherwise}},\end{array}\right.
			\end{aligned}
		\end{equation}
		otherwise, define 
		\begin{equation}
			\label{LPsolutions2}
			\begin{aligned}
				x^*_{v,1}(k)  = & \left\{\begin{array}{ll}y^*_v, &
					{\text{if}}~v = \sigma_{1}(k), \ldots, \sigma_{t(k)}(k);~ \\
					1-\sum_{g=1}^{t(k)} y_{\sigma_g(k)}^*,&\text{if}~v= \sigma_{t(k)+1}(k);\\
					0, &{\text{otherwise}},\end{array}\right.
			\end{aligned}
		\end{equation}
		By the construction, it is not difficult to verify that $(\boldsymbol{x}^*, \boldsymbol{y}^*)$ is a feasible solution of problem \eqref{SC2-LP}. 
		The proof of the equivalence of problems \eqref{SC2} and \eqref{SC1} is similar. }
		Problem \eqref{SC1} is the well-known set covering problem and problem \eqref{LP1} is its LP relaxation.
		From \cite[Pages 110-112]{Vazirani2001}, the integrality gap of the set covering problem is $\Omega(\log (|\mathcal{K}|))$, which can go to infinity as $|\mathcal{K}|\rightarrow \infty$.

		\revv{
			Next, we prove case (ii).
			The dual of problem \eqref{LP1} can be written as
			\begin{align}
				%\begin{aligned}
				\max_{\substack{\boldsymbol{w}}} & ~~  \sum_{k \in \mathcal{K}}w_k \nonumber \\
				~~{\text{s.t.~}} &~~    \sum_{k\,:\,v \in \mathcal{V}(k)} w_k \leq 1,~\forall~v \in \mathcal{V},  \nonumber                                                         \\
				& ~~    w_k \geq 0,\,\forall~k\in\mathcal{{K}}.
				\label{LP2}
				%\tag{SC-LP}
				%\end{aligned}
			\end{align}
			Let $\boldsymbol{y}^*$ and $\boldsymbol{w}^*$ be the optimal solutions to the primal LP \eqref{LP1} and the dual LP \eqref{LP2}, respectively. 
			By strong duality, we have $\sum_{v \in \mathcal{V}}y_v^* = \sum_{k \in \mathcal{K}}w_k^*$. 
			Let
			\begin{equation}
				\label{Sdef}
				\mathcal{S} = \{ v \in \mathcal{V} \mid y_v^*  > 0 \}.
			\end{equation} 
			Then the point $\bar{\boldsymbol{y}}$ defined by
			\begin{equation}
				\label{LPsolutionsy}
				\begin{aligned}
					\bar{y}_{v}= & \left\{\begin{array}{ll}1, &
						~ {\text{if}}~i \in \mathcal{S};\\
						0, & ~{\text{otherwise}},\end{array}\right.
				\end{aligned}~\forall~v \in \mathcal{V},
			\end{equation} 
			is feasible for problem \eqref{SC1}, showing that the optimal value of problem \eqref{SC1} is at most $|\mathcal{S}|$.
			On the other hand, observe that 
			\begin{equation}
				\label{tmpineq1}
				\begin{aligned}
					|\mathcal{S}| & = && \sum_{v \,: \, y_v^*> 0 } 1 \\
					&= && \sum_{v \, :\,  y_v^*> 0 }  \left(\sum_{k \, :\, v \in \mathcal{V}(k)} w^*_k\right)\\ 
					& \leq && \sum_{v\in \mathcal{V}}  \left(\sum_{k \, :\,  v \in \mathcal{V}(k)}w^*_k\right)\\
					& = && \sum_{k\in \mathcal{K}}  \sum_{v \in \mathcal{V}(k)}w^*_k\\
					& = && \sum_{k \in \mathcal{K}}|\mathcal{V}(k)| w_k^*\leq u \sum_{\mathcal{K}}\omega_k^*,
				\end{aligned}
			\end{equation}
			where the second equality follows from the complementary slackness conditions, i.e., if $y^*_v > 0$, $\sum_{k\, : \, v\in \mathcal{V}(k)}w_k^*=1$ must hold, and the last inequality follows from $|\mathcal{V}(k)| \leq u$ for all $k \in \mathcal{K}$.
			This shows that the integrality gap in  this case is at most  $\frac{|\mathcal{S}|}{\sum_{\mathcal{K}}\omega_k^*}\leq u	$.
		}

		\revv{Now, we prove case (iii).}
		In this case, summing up constraints \eqref{nodecapcons} and using \eqref{noprocesscons} and \eqref{onlyonenode}, we have 
		\begin{align*}
		\sum_{v\in \mathcal{V}}y_v & ~ \geq \sum_{v\in \mathcal{V}}\frac{1}{\mu_v}\sum_{k \in \mathcal{K}} \lambda_1(k) x_{v,1}(k) \\
		& ~= \sum_{k \in \mathcal{K}}\sum_{v\in \mathcal{V}(k)}\frac{1}{\mu_v} \lambda_1(k) x_{v,1}(k) \\
		& ~\geq \frac{1}{t}\sum_{k \in \mathcal{K}} \sum_{v \in \mathcal{V}(k)}x_{v,1}(k) = \frac{1}{t}|\mathcal{K}|,
		\end{align*}
		which implies that the objective of relaxation \eqref{LP1} is at least $\frac{1}{t}|\mathcal{K}|$.
		This, together with the fact that the objective value of problem \eqref{SC} is at most $|\mathcal{K}|$ (as each activated node must host at least one function), shows that the integrality gap in this case is at most $t$. 
		
			\revv{Finally, we prove case (iv).}
			In this case, problem \eqref{LP} reduces to 
			\begin{align}
			%\begin{aligned}
			\min_{\substack{\boldsymbol{x},\,\boldsymbol{y}}} &~~   \sum_{v \in \mathcal{V}}y_v \nonumber \\
			~~{\text{s.t.~}} &~~   \eqref{onlyonenode},~\eqref{noprocesscons},~ \sum_{k \in \mathcal{K}}x_{v,1}(k) \leq y_v,~\forall~k \in \mathcal{K}, \nonumber                                                         \\
			&~~   x_{v,1}(k),~y_v\in[0,1],\,\forall~v\in\mathcal{{V}}, ~k\in\mathcal{K}.
			\label{MC}
			%\tag{MC}
			%\end{aligned}
			\end{align}
			%Problem \eqref{MC} is equivalent to finding a maximum matching on a bipartite graph $\mathcal{G}=\{\mathcal{K}\cup\mathcal{V}, \mathcal{L}'\}$ ($\mathcal{K}$ represents the set of functions $\{f_1^k\}$, $\mathcal{V}$ represents the set of cloud nodes, and $\mathcal{L}'$ includes all links $(k,v)$ with $k \in \mathcal{K}$ and $v\in\mathcal{V}(k) \subseteq \mathcal{V}$).
			%In particular, if the above problem is feasible, the objective value must be equal to $|\mathcal{K}|$.
			One can check that the constraint matrix in problem \eqref{MC} is a totally unimodular matrix, showing that the integrality gap of this problem is zero and the original problem \eqref{newmip} can be solved to global optimality; see \cite[Chapter 4]{Conforti2014}.
\end{proof}

%This shows the integrality gap of relaxation \eqref{LP1} can be infinite, which provides more insight into the intrinsic difficulty of the original problem.

\section{An LP \revv{dynamic} rounding-and-refinement algorithm}
\label{sec:heur}

In this section, we will focus on designing an efficient algorithm, based on the proposed LP relaxation \eqref{lp}, to obtain a high-quality solution for problem \eqref{newmip}.
The basic idea of the proposed algorithm is to decompose the hard problem \eqref{newmip} into two relatively easy subproblems and solve two subproblems separately while taking their connection into account.
Specifically, in the first stage, we find a binary vector $(\bar{\boldsymbol{x}},\bar{\boldsymbol{y}})$ for the VNF placement subproblem (i.e., $(\bar{\boldsymbol{x}},\bar{\boldsymbol{y}})$ satisfying constraints \eqref{onlyonenode}-\eqref{nodecapcons}) using an LP \revv{dynamic} rounding procedure, which takes traffic routing into account.
In the second stage, based on the binary vector $(\bar{\boldsymbol{x}},\bar{\boldsymbol{y}})$, we use an LP refinement procedure to solve the traffic routing subproblem to obtain a solution that satisfies the E2E delay constraints \eqref{maxdelay1}--\eqref{delayconstraint} of all services.
%The discussion of the LP rounding-and-refinement algorithm is based on the novel LP relaxation \eqref{lp}.
%However, we can also use the natural LP relaxation \eqref{milplp} instead.
%\vspace*{0.1cm}\\
%%
%{\bf\noindent$\bullet$ Solving the VNF Placement Subproblem\vspace{0.1cm}}
\subsection{Solving the VNF Placement Subproblem}
\label{VNFalg}

In this subsection, we solve the VNF placement subproblem by constructing a binary vector $(\bar{\boldsymbol{x}}, \bar{\boldsymbol{y}})$ that satisfies constraints \eqref{onlyonenode}--\eqref{nodecapcons}.
Since vector $\bar{\boldsymbol{y}}$ can be uniquely determined by vector $\bar{\boldsymbol{x}}$, in the following we concentrate on constructing the binary vector $\bar{\boldsymbol{x}}$.
To do this, we first solve the LP relaxation   \eqref{lp}, denoted its solution by $(\boldsymbol{x}^*, \boldsymbol{y}^*, \boldsymbol{r}^*,\boldsymbol{\theta}^*)$.
If $\boldsymbol{x}^*$ is a binary vector, we obtain a feasible solution $\bar{\boldsymbol{x}}:=\boldsymbol{x}^*$ for the VNF placement subproblem.
Otherwise, we set $x_{v,s}(k)=1$ in relaxation \eqref{lp} if $x^*_{v,s}(k)=1$.
Then we choose one variable, denoted as  $x_{v_0,s_0}(k_0)$, whose value $x^*_{v_0,s_0}(k_0)$ is the largest among the remaining variables, i.e.,
\begin{align}
	x^*_{v_0,s_0}(k_0)=  \max \big \{  x^*_{v,s}(k) \, : \, 0 < x^*_{v,s}(k)< 1, ~ v \in \mathcal{V}, \nonumber \\
	\qquad k\in \mathcal{K},~s\in \mathcal{F}(k)\big\}. \label{maxvalue}
\end{align}
Next we decide to round variable $x_{v_0,s_0}(k_0)$ to one or zero.
In particular, we first set $x_{v_0,s_0}(k_0)=1$ in problem \eqref{lp}.
If the above modified LP is infeasible, we set $x_{v_0, s_0}(k_0) =0$ and continue to round other variables respect to the values $\{x^*_{v,s}(k)\}$.
\revv{If the above modified LP is feasible, we update the LP relaxation solution and} repeat the above procedure to the solution of the modified LP until a binary solution is obtained.
The details are summarized in the following Algorithm \ref{roundingx}.
\begin{algorithm}[!h]
	\caption{{An LP \revv{dynamic} rounding procedure for solving the VNF placement subproblem}}
	\label{roundingx}
	\begin{algorithmic}[1]
		{
		\renewcommand{\algorithmicrequire}{\textbf{Input:}}
		\renewcommand{\algorithmicensure}{\textbf{Output:}}
		\STATE Initialize the set $\mathcal{A}=\varnothing$;
		\STATE Solve problem \eqref{lp} to obtain its solution $(\boldsymbol{x}^*,\boldsymbol{y}^*,\boldsymbol{r}^*,\boldsymbol{\theta}^*)$;
		\WHILE {(there exists some $v \in \mathcal{V}$, $k\in \mathcal{K}$, and $s\in \mathcal{F}(k)$ such that $0 < x^*_{v,s}(k) < 1 $)}
		\STATE For each $v \in \mathcal{V}$, $k\in \mathcal{K}$, and $s\in \mathcal{F}(k)$ with $x^*_{v,s}(k)=1$, if constraint $x_{v,s}(k)=1$ is not in set $\mathcal{A}$, add it into set $\mathcal{A}$;
		\STATE Compute $(v_0, s_0,k_0)$ such that \eqref{maxvalue} holds and add constraint $x_{v_0,s_0}(k_0)=1$ into set $\mathcal{A}$;
		\STATE Add the constraints in set $\mathcal{A}$ into problem \eqref{lp} to obtain a modified LP;
		\IF {(the modifed LP problem is feasible)}
		\STATE Set $(\boldsymbol{x}^*,\boldsymbol{y}^*,\boldsymbol{z}^*,\boldsymbol{\theta}^*) \leftarrow $ the optimal solution of the modifed LP problem;
		\ELSE
		\STATE Replace constraint $x_{v_0,s_0}(k_0)=1$ by constraint $x_{v_0,s_0}(k_0)=0$ in set $\mathcal{A}$ and set  $x^*_{v_0,s_0}(k_0)\leftarrow 0 $;
		\ENDIF
		%				\STATE If the modifed LP is feasible, let $(\boldsymbol{x}^*,\boldsymbol{y}^*,\boldsymbol{z}^*,\boldsymbol{\theta}^*) $ be its solution; otherwise, replace constraint $x_{v_0,s_0}(k_0)=1$ by constraint $x_{v_0,s_0}(k_0)=0$ in set $\mathcal{A}$ and set  $x^*_{v_0,s_0}(k_0)\leftarrow 0 $;
		\ENDWHILE
		%		\IF {($(x^*,y^*)$ satisfies constraints \eqref{onlyonenode}-\eqref{nodecapcons})}
		%		\STATE Set $\bar{x}\leftarrow x^*$ and $\bar{y}\leftarrow y^*$;
		%		\ELSE
		%		\STATE Report that the algorithm cannot find a feasible solution for the VNF placement problem; 
		%		\ENDIF
		\STATE If vector $(\boldsymbol{x}^*,\boldsymbol{y}^*)$ satisfies constraints \revvv{\eqref{onlyonenode}--\eqref{nodecapcons}}, then the binary vector $(\bar{\boldsymbol{x}},\bar{\boldsymbol{y}})\leftarrow(\boldsymbol{x}^*,\boldsymbol{y}^*)$ is feasible for the VNF placement subproblem; otherwise declare that the algorithm fails to find a feasible solution.
		}
		%\RETURN $x$ and $y$.
	\end{algorithmic}
\end{algorithm}

Two remarks on the above rounding strategy are as follows. 
First, the above rounding strategy makes sure that we can round one variable, taking a fractional value at the current solution, at a time and more importantly this variable can be rounded to a binary value that is consistent to other already rounded variables.
This is in sharp contrast to the algorithm in \cite{Chowdhury2012} where the variables are rounded \revv{in a one-shot fashion} without ensuring the consistency of the current rounding variable with other already rounded variables.
%%
%In sharp contrast, the algorithm in \cite{Chowdhury2012} rounds variables without ensuring the consistency of the current rounding variable with the rounded variables.
%
%In practice, the (initial) solution $(\boldsymbol{x}^*,\boldsymbol{y}^*,\boldsymbol{z}^*,\boldsymbol{\theta}^*)$ of the LP relaxation \eqref{lp} in Algorithm \ref{roundingx} plays an important role in finding a feasible solution for the VNF placement problem.
%Indeed, a considerable large number of the variables $\boldsymbol{x}$ take values at one at this solution (as observed in our experiments) and these variables will be fixed to one in step 3 of the first iteration.
%On the other hand, 
Second, our rounding strategy takes traffic routing into consideration as the modified LP contains the information of traffic routing of all services. 
Indeed, when setting $x_{v_0,s_0}(k_0)=1$ in problem \eqref{lp}, there may not exist traffic routing strategies for some services (due to the limited node/link capacities or the E2E latency thresholds).
In this case, the modified LP problem is likely to be infeasible, and we exploit this by re-setting $x_{v_0,s_0}(k_0)=0$ in modified LP problem \eqref{lp}; see steps 7--11 for more details.
\subsection{Solving the Traffic Routing Subproblem}
Once we get a binary vector $(\bar{\boldsymbol{x}}, \bar{\boldsymbol{y}})$, we still need to solve the traffic routing subproblem that finds paths connecting two adjacent functions and satisfying the E2E delay of all services.
This can be done by solving problem \eqref{newmip} with $\boldsymbol{x}=\bar{\boldsymbol{x}}$ and $\boldsymbol{y}=\bar{\boldsymbol{y}}$.
In this case, the objective function in problem \eqref{newmip} reduces to $$g(\boldsymbol{\theta})=\sum_{k \in\mathcal{K}} \theta_L(k).$$
Similarly, we solve the LP problem \eqref{lp} with $\boldsymbol{x}=\bar{\boldsymbol{x}}$ and $\boldsymbol{y}=\bar{\boldsymbol{y}}$ to obtain a solution $(\boldsymbol{r}^*, \boldsymbol{\theta}^*)$.
Due to the (possibly) fractional values of $\{ r^*_{ij}(k,s,1) \}$, $\theta^*(k,s)$ can be smaller than the communication delay incurred by the traffic flow from the node hosting function $f_s^k$ to the node hosting function $f_{s+1}^{k}$.
To recompute the communication delay based on solution $(\boldsymbol{r}^*, \boldsymbol{\theta}^*)$, we need to solve the NP-hard {Min-Max-Delay} problem \cite{Liu2017a}.
Fortunately, there exists an efficient polynomial time ($1+\epsilon$)-approximation algorithm for this problem \cite{Liu2017a}.
%
%Due to space reason, we do not discuss 
%
After recomputing the communication delays between all pairs of nodes hosting two adjacent functions, we can compute the total delay of each service $k$, denoted as $\bar{\theta}(k)$.
If $\bar{\theta}(k)> \Theta(k)$ for some service $k$, the current routing strategy is infeasible as it violates the E2E delay constraint of service $k$.
We then use an iterative LP refinement procedure to try to get a solution that satisfies the E2E delay constraints of all services.

The idea of our refinement procedure is to increase the weights of the variables $\theta_{L}(k) $ corresponding to the service whose E2E delay constraint is not satisfied at the current solution, in order to refine the solution.
In particular, we change the objective function $g(\boldsymbol{\theta})$ in problem \eqref{lp} into $$\hat{g}(\boldsymbol{\theta})=\sum_{k \in \mathcal{K}} \omega_k \theta_L(k),
$$
where $\omega_k \geq 1 $ for all $k \in \mathcal{K}$.
At each iteration, we solve problem \eqref{lp} (with the objective function $\hat{g}(\boldsymbol{\theta})$, $\boldsymbol{x}=\bar{\boldsymbol{x}}$, and $\boldsymbol{y}=\bar{\boldsymbol{y}}$) to obtain its solution $(\boldsymbol{r}^*, \boldsymbol{\theta}^*)$.
If, for some service $k$, the E2E delay constraint is violated at this solution, we increase $\omega_k$ by a factor of $\rho > 1$, and solve problem \eqref{lp} again.
The procedure is repeated until the solution satisfies the E2E delay constraints of all services or the iteration number reaches a predefined parameter IterMax.
We summarize the above procedure in Algorithm \ref{al:routing}.
%\blue{Notice that as stated in Section II-B-2, $P$ can be set as a sufficiently large number, and hence we can always decompose the traffic flow $\{r_{ij}(k,s,1)\}$ into a suitable number of paths.}
%\vspace{0.1cm}\\

\begin{algorithm}[t]
	\caption{{An iterative LP refinement procedure for solving the traffic routing subproblem}}
	\label{al:routing}
	\begin{algorithmic}[1]
		\renewcommand{\algorithmicrequire}{\textbf{Input:}}
		\renewcommand{\algorithmicensure}{\textbf{Output:}}
		{\STATE Set $\rho >1$, $\text{IterMax} \geq 1$, $t=0$, and $\omega_k = 1$ for $k \in \mathcal{K}$;
			\WHILE {$t < \text{IterMax}$}
			\STATE Solve problem \eqref{lp} (with the objective function $\hat{g}(\boldsymbol{\theta})$, $\boldsymbol{x}=\bar{\boldsymbol{x}}$, and $\boldsymbol{y}=\bar{\boldsymbol{y}}$) to obtain its solution $(\boldsymbol{r}^*, \boldsymbol{\theta}^*)$;
			\STATE For each $k\in \mathcal{K}$, compute the total delay $\bar{\theta}(k)$ based on $(\boldsymbol{r}^*, \boldsymbol{\theta}^*)$ and  $(\bar{\boldsymbol{x}}, \bar{\boldsymbol{y}})$;
			%				\STATE
			%				If $\bar{\theta}(k) \leq \Theta_k$ for all $k \in \mathcal{K}$, we stop with the feasible solution $(\boldsymbol{z}^*, \boldsymbol{\theta}^*)$; otherwise, for each $k\in \mathcal{K}$ with $\bar{\theta}(k) > \Theta_k$, we update $\omega_k \leftarrow \rho \omega_k$; 
			\IF{($\bar{\theta}(k) \leq \Theta_k$ for all $k \in \mathcal{K}$)}
			\STATE Stop with the feasible solution $(\boldsymbol{r}^*, \boldsymbol{\theta}^*)$;
			\ELSE
			\STATE For each $k\in \mathcal{K}$ with $\bar{\theta}(k) > \Theta_k$,  set $\omega_k \leftarrow \rho \omega_k$;
			\ENDIF
			\STATE Set $t\leftarrow t+ 1$;
			\ENDWHILE
		}
		%		\RETURN $(r^*,)$
	\end{algorithmic}
\end{algorithm}%
%{\bf\noindent$\bullet$ Complexity Analysis\vspace{0.1cm}}
\subsection{\revv{Analysis Results}}
\label{problemcomplexity}

\revv{In this subsection, we present some theoretical analysis results of the proposed \LPRR algorithm, including its complexity and approximation ratio (in some special cases). }

\subsubsection{\revv{Complexity Analysis}}
The dominant computational cost of our algorithm is to solve the LP problems in form of \eqref{lp}. 
Indeed, the number of solving problems \eqref{lp} in Algorithms \ref{roundingx} and \ref{al:routing} is upper bounded by $|\mathcal{V}|\sum_{k \in \mathcal{K}}\ell_k$ and  $\text{IterMax}$, respectively.
By summing up the upper bound of solved LPs in Algorithms \ref{roundingx} and \ref{al:routing}, we obtain the upper bound of solved LPs in the proposed \LPRR algorithm, which are 
\begin{equation}\label{upperLP}
	|\mathcal{V}|\sum_{k \in \mathcal{K}}\ell_k + \text{IterMax}.
\end{equation}
%
%Since an LP can be solved using the (polynomial-time) interior-point method \cite{Renegar1988}, it follows that the worst-case complexity of our proposed \LPRR~algorithm is polynomial.
%The latter is computed by summing up the upper bound of solved LPs in Algorithms \ref{roundingx} and \ref{al:routing}, which are $|\mathcal{V}|\sum_{k \in \mathcal{K}}\ell_k$ and $\text{IterMax}$, respectively.
%%
The worst-case complexity of solving an LP is approximately $\mathcal{O}((n+m)^{1.5}n^2)$ using the interior-point method in \cite[Section 6.6.1]{Ben-Tal2001}, where $n$ and $m$ are the numbers of variables and constraints, respectively.
This, together with the fact that the numbers of variables and constraints in problem \eqref{lp} are $\mathcal{O}(|\mathcal{L}|\sum_{k\in \mathcal{K}}\ell_k)$ and $\mathcal{O}(\min\{|\mathcal{I}|\sum_{k\in \mathcal{K}}\ell_k, |\mathcal{L}|\})$, respectively, implies that problem \eqref{lp} can be solved in 
$$
\small
\mathcal{O}\left (\left (|\mathcal{V}|\sum_{k \in \mathcal{K}}\ell_k+\text{IterMax}\right)*\left(|\mathcal{L}|\sum_{k\in \mathcal{K}}\ell_k\right)^{3.5}\right).
$$
In sharp contrast, the worst-case complexity of using the standard MILP solvers like Gurobi \cite{Gurobi} to solve problem \eqref{newmip} is exponential, which makes it unsuitable for solving large-scale problems.

\subsubsection{\revv{Approximation Ratio Analysis}}

\revv{Next, we present some approximation ratio analysis results of the proposed \LPRR algorithm for the special cases considered in Section \ref{integralitygapanalysis} (i.e., cases (i)--(iv) in Theorem \ref{integralitygap}) in the following theorem.
%The results are summarized in the following theorem.
}

\begin{theorem}
	\revv{Consider the cases in Theorem \ref{integralitygap}. Then, for  cases (i) and (iii), any polynomial time algorithm including the proposed \LPRR algorithm cannot return a solution with a constant approximation ratio guarantee for problem \eqref{newmip};
	for cases (ii) and (iv), the \LPRR algorithm can  return a $u$-approximation solution and a global solution for problem \eqref{newmip}, respectively.}
\end{theorem}

\begin{proof}
	\revv{We prove the results in the four cases (i)--(iv), separately.}
	
\revv{
	(i) $\mu_v = \infty$ for all $v \in \mathcal{V}$.
	As shown in the proof of Theorem \ref{integralitygap}, problem \eqref{newmip} with $\mu_v = \infty$ for all $v \in \mathcal{V}$ reduces to a generic set covering problem \eqref{SC1}, which does not admit a polynomial time constant approximation ratio algorithm \cite{Dinur2014}. 
	As a result, the proposed \LPRR algorithm cannot return a solution with a constant approximation ratio guarantee.
}

\revv{(ii) $\mu_v = \infty$ for all $v \in \mathcal{V}$ and $|\mathcal{V}(k)| \leq u$ for all $k \in \mathcal{K}$.
	As shown in Section \ref{integralitygapanalysis}, (i) problem \eqref{newmip} in this case reduces to problem \eqref{SC2}; and (ii) if $\bm{y}^*$ is an optimal solution of problem \eqref{LP1}, 
	then $(\bm{x}^*,\bm{y}^*)$ is an optimal solution of problem \eqref{SC2-LP} where $\bm{x}^*$ is defined in  \eqref{LPsolutions} or \eqref{LPsolutions2}.
	It is simple to verify that applying the \LPRR algorithm to solve problem \eqref{SC2} is equivalent to rounding $y_v^*$ to one in the descending order of ${y_v^*}$ until a feasible solution of problem \eqref{SC2} is obtained.
	%(with $(\bm{x}^*,\bm{y}^*)$ being its LP relaxation solution where $\bm{x}^*$ is defined in \eqref{LPsolutions} or \eqref{LPsolutions2}) is equivalent to applying the rounding procedure on $\bm{y}^*$ to obtain a feasible solution of problem \eqref{SC1} (i.e., rounding $y^*_v$ to one according to the decreasing values of $\{y_v^*\}$).
	In this way, we will obtain a feasible solution $\bar{\bm{y}}$ given in \eqref{LPsolutionsy} of problem \eqref{SC1} with its optimal value at most $|\mathcal{S}|$ where $\mathcal{S} $ is defined \eqref{Sdef}.
	%; see the proof of case (ii) of Theorem \ref{integralitygap} for more details.
	By \eqref{tmpineq1} and the fact that the optimal value of problem \eqref{SC1} is larger than or equal to that of its relaxation \eqref{LP1}, the optimal value of problem \eqref{SC1} must be at least $\frac{|\mathcal{S}|}{u}$.
	Consequently, the \LPRR algorithm can return a $u$-approximation solution for problem \eqref{newmip} in this case.
}

\revv{(iii) $\lambda_1(k) \geq \frac{1}{t}\mu_v$ for all $k \in \mathcal{K}$ and $v \in \mathcal{V}(k)$.
	In this case, checking the feasibility of problem \eqref{SC} with $\lambda_1(k) \geq \frac{1}{3}\mu_v$ for all $k \in \mathcal{K}$ and $v \in \mathcal{V}(k)$ is strongly NP-complete; see Section II of \cite{Chen2021d} for a rigorous proof.  
	Consequently, there does not exist a polynomial time constant approximation ratio algorithm to solve problem \eqref{SC} with $t=3$, and hence the proposed \LPRR algorithm cannot return a solution with a constant approximation ratio guarantee.}

\revv{(iv) Each node can host at most one function. 
	As shown in Section \ref{integralitygapanalysis}, the LP relaxation of problem \eqref{newmip} reduces to \eqref{MC} and solving \eqref{MC} will return an integral solution $(\bar{\bm{x}}, \bar{\bm{y}})$.
	Therefore, the proposed \LPRR algorithm can return a global solution of problem \eqref{newmip} in this case.
}
\end{proof}

\subsection{\revv{Comparison with Existing Algorithms}}
\label{comexisting}

\revv{
	In this subsection, we carefully compare the proposed LPdRR algorithm with existing algorithms in the literature \cite{Chowdhury2012}, \cite{Yu2008}, and \cite{Zhang2012}.
	\\[2pt]
	{\bf $\bullet$ Comparison with non-LPR types of algorithms in \cite{Yu2008} and \cite{Zhang2012}\\[2pt]}
	{\indent We first compare the proposed \LPRR algorithm with the inflation optimization (\GREEDY) and \PSO algorithms, which are adapted from those in \cite{Yu2008} and \cite{Zhang2012}, respectively.}
	{\GREEDY maps VNFs into cloud nodes using an inflation optimization procedure while \PSO maps VNFs into cloud nodes using a particle swarm optimization procedure.}
	{These two algorithms then find paths connecting two adjacent VNFs in the network by solving a multicommodity flow problem.}
	{Due to the space limit, we provide the detailed description of the \GREEDY and \PSO algorithms in Section III of \cite{Chen2021d}.} 
	%The considered network slicing problem is to jointly determine (i) virtual network function (VNF) mapping  (which maps VNFs into cloud nodes in the network) and (ii)  virtual link mapping (which finds paths connecting two adjacent VNFs of each service in the network).
	%This leads to two subproblems: one for VNF mapping and the other for virtual link mapping.
	
	Compared to the \GREEDY and \PSO algorithms that address the VNF and link mappings separately without a coordination, the proposed \LPRR algorithm can coordinate VNF and link mappings.
	This is because that the LP relaxation of the MILP formulation, used for finding a solution for the VNF placement, also considers the traffic routing of all services. 
	Consequently, the proposed \LPRR algorithm is likely to return a better solution than \GREEDY and \PSO; see \cite{Chowdhury2012,Fischer2013} for a similar discussion on the advantage of LPR types of algorithms.}\\[2pt]
\revv{{\bf $\bullet$ Comparison with the \LPR algorithm in \cite{Chowdhury2012}\\[2pt]}}
\revv{\indent Next, we compare the proposed algorithm with the LP one-shot rounding (\LPR) algorithm, which is adapted from the one in \cite{Chowdhury2012}. 
	\LPR first solves the compact LP relaxation of the network slicing problem to obtain its solution and then uses a rounding strategy to map VNFs into cloud nodes based on this solution.
	Similar to that in \GREEDY and \PSO, \LPR then finds paths connecting two adjacent VNFs in the network by solving a multicommodity flow problem.
	A detailed description of the \LPR can be found in Section III of \cite{Chen2021d}.
	
	The key difference between \LPR and \LPRR lies in their rounding strategy.
	In particular, compared with \LPR, a much more stronger coordination between VNF and link mappings is explored in the proposed \LPRR algorithm; see  the last paragraph of Section \ref{VNFalg} for a detailed discussion.
	%First, the rounding strategy in the proposed \LPRR makes sure that we can round one variable, taking a fractional value at the current solution, at a time and more importantly this variable can be rounded to a binary value that is consistent with other already rounded variables.
	%
	%This is in sharp contrast to the LPR algorithm in \cite{Chowdhury2012} where the variables are rounded without ensuring the consistency of the current rounding variable with other already rounded variables.
 %Second, our rounding strategy takes traffic routing into consideration as the modified LP contains the information of traffic routing of all services. 
%	Indeed, when setting $x_{v_0,s_0}(k_0)=1$ in problem (LP-II), there may not exist traffic routing strategies for some services (due to the limited node/link capacities or the E2E latency thresholds).
	%In this case, the modified LP problem is likely to be infeasible, and we exploit this possibility by re-setting $x_{v_0,s_0}(k_0)=0$ in modified LP problem (LP-II); see steps 7-11 for more details.
	This advantage enables that the proposed \LPRR algorithm can generally find better solutions than the \LPR algorithm in \cite{Chowdhury2012}.
	We remark that a modified version of the proposed \LPRR can always outperform \LPR.
	More specifically, if we remove step 8 of Algorithm \ref{roundingx} (i.e., the step that dynamically updates
	the relaxation solution), the resultant algorithm will do rounding based on the static LP solution and is thus called LP static rounding-and-refinement algorithm (\LPRRSS).
	A detailed description of the \LPRRSS can be found in Section IV of \cite{Chen2021d}.
	%\LPRRSS is indeed {equivalent to the \LPR algorithm with a refinement procedure to guarantee the consistency of the rounded value of the current rounding variable}; see Section III of \cite{Chen2021d} for a detailed description of the resultant algorithm.
	The resultant \LPRRSS can be guaranteed to find feasible solutions as long as \LPR can find feasible solutions for the network slicing problem while the converse is not true.
	However, as will be shown in Section \ref{compofalgs}, dynamically updating the relaxation solution in step 8 of Algorithm \ref{roundingx} is helpful in (i) finding a feasible solution with a smaller number of activated cloud nodes, and (ii) improving the solution efficiency of the algorithm.
}

\section{Simulation Results}
\label{sec_num}

In this section, we present simulation results to illustrate the effectiveness and efficiency of the proposed LP relaxation \eqref{lp} and the proposed \LPRR algorithm for solving the network slicing problem.
Specifically, we first perform numerical experiments to compare the performance of solving the natural LP relaxation \eqref{milplp} and the proposed LP relaxation \eqref{lp}.
Then, we present some simulation results to demonstrate the efficiency and effectiveness of our proposed \LPRR algorithm over the state-of-the-art approaches \revv{in \cite{Chen2020},  \cite{Chowdhury2012}, \cite{Yu2008}, and \cite{Zhang2012}}.

In problem \eqref{newmip}, we choose $\sigma  = 0.001$ and $ P=2 $, as suggested in \cite{Chen2020}.
In Algorithm \ref{al:routing}, we set $\rho = 5$ and $\text{IterMax}=10$.
\revv{We use Gurobi 9.5 \cite{Gurobi} to solve all MILP 
and LP problems. 
It is worthwhile noting that Gurobi is one of the fastest solvers for solving MILP and LP problems \cite{Mittelmann2021}.}
When solving the MILP problems, the time limit is set to be 1800 seconds, and the relative gap tolerance is set to be 0.1\%, i.e., a feasible solution which has an optimality gap of 0.1\% is considered to be optimal.
All experiments were performed on a server with 2 Intel Xeon E5-2630 processors and 98 GB of RAM, using the Ubuntu GNU/Linux Server \revv{20.04} x86 64 operating system.

\rev{We test all algorithms on the fish network topology studied in \cite{Zhang2017} and the JANOS-US, NOBEL-GERMANY, and POLSKA network topologies taken from the SNDLIB \cite{Orlowski2010}}. 
\rev{Due to the space limit, we only present the simulation results on the fish network topology in this paper and more simulation results on other network topologies can be found in Section \revv{V} of \cite{Chen2021d}.}
\rev{The fish network topology contains 112 nodes and 440 links, including 6 cloud nodes.}
%
%Among these 112 nodes, 86 nodes can be potentially chosen as the source node of the flows, 1 node can be chosen as the destination node of the flows, and 6 nodes are cloud nodes that can potentially process service functions.
%Among the 6 cloud nodes, 5 cloud nodes are randomly chosen to process two service functions of $ \{f^1, \ldots, f^4\} $ and the remaining one is chosen to process all the service functions. 
%
The cloud nodes' and links' capacities are randomly generated within $ [50,100] $ and $ [5,55] $, respectively.
The NFV	 and communication delays on the cloud nodes and links are randomly generated within $\{3,4,5,6\}$ and $\{1,2\}$, respectively.
For each service $k$, node $S(k)$ is randomly chosen from the available nodes and node $D(k)$ is set to be the common destination node; SFC $ \mathcal{F}(k) $ is a sequence of functions randomly generated from $ \{f^1, \ldots, f^4\} $ with $ |\mathcal{F}(k)|=3 $; $ \lambda_s(k) $'s are the service function rates which are all set to be the same integer value, randomly generated within $ [1,11] $; $ \Theta_k $ is set to $ 20+(3*\text{dist}_k+\alpha) $ where $ \text{dist}_k $ is the delay of the shortest path between nodes $ S(k) $ and $ D(k) $ and $ \alpha $ is randomly chosen in $[0,5]$.
The above parameters are carefully chosen to make sure that the constraints in problem \eqref{newmip} are neither too tight nor too loose.
For each fixed number of services, 100 problem instances are randomly generated and the results presented below are obtained by averaging over these problem instances.

\subsection{Comparison of LP Relaxation \eqref{milplp} and Proposed LP Relaxation \eqref{lp}}
\label{compofLPs}
In this subsection, we compare the performance of solving the natural LP relaxation \eqref{milplp} and the proposed LP relaxation \eqref{lp}.
%

%For each problem instance, we define the CPU time ratio as follows:
%\begin{equation*}
%	{\text{CPU time ratio}} = \frac{\text{T(LP-I)}}{\text{T(LP-II)}},
%	%	\text{ratio} = \frac{\text{T}\eqref{mip}}{\text{T}\eqref{newmip}},
%\end{equation*}
%where $\text{T(LP-I)}$ and $\text{T(LP-II)}$ are the CPU time taken by solving relaxations \eqref{milplp} and \eqref{lp}, respectively.
%The {CPU time ratio} reflects the computational efficiency of solving relaxation \eqref{lp} over that of solving relaxation \eqref{milplp}, i.e., the larger the {CPU time ratio}, the higher efficiency of solving relaxation \eqref{lp} (as compared with solving relaxation \eqref{milplp}).

\begin{figure}[t]
	\centering
	\includegraphics[height=\figuresize]{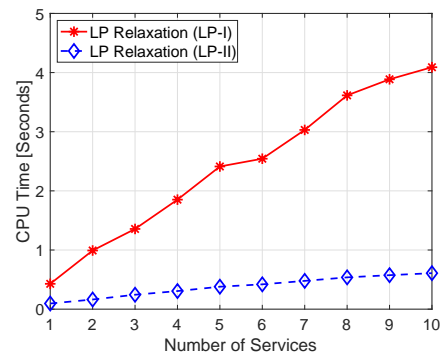}
	\caption{\rev{The CPU time taken by solving relaxations (LP-I) and (LP-II).}}
	\label{LPtime}
\end{figure}
We first compare the computational efficiency of solving two relaxations \eqref{milplp} and \eqref{lp}.
\rev {Fig. \ref{LPtime} plots the {CPU time} versus the numbers of services.
	From Fig. \ref{LPtime}, it can be clearly seen that it is much more efficient to solve relaxation \eqref{lp} than relaxation \eqref{milplp}.
	In particular, when the number of services is equal to 10, the CPU time of solving relaxation \eqref{lp} is about \revv{0.6} seconds while that of solving relaxation \eqref{milplp} is more than \revv{4.1} seconds.
	%In some cases ($|\mathcal{K}|\geq 5$), solving relaxation \eqref{lp} is even 10+ times faster than solving relaxation \eqref{milplp}.
	%
	In addition, we can observe from Fig. \ref{LPtime} that,  with the increasing number of services, the {CPU time} of \eqref{milplp}  increases much faster, as compared with that of relaxation \eqref{lp}.
	This is mainly due to the fact that the numbers of variables and constraints in relaxation \eqref{milplp} increase much faster than those in relaxation \eqref{lp} as the number of services increases.}

%\begin{figure}[t]
%	\centering
%	\includegraphics[height=\figuresize]{LPtimeratio}
%	\caption{{The ratio of the CPU time taken by solving relaxations \eqref{milplp} and \eqref{lp}.}}
%	\label{LPtimeratio}
%\end{figure}
%
%%For each fixed number of services, we also compute the geometric mean of the ratio over the 100 instances.
%Fig. \ref{LPtimeratio} plots the {CPU time ratio} versus the numbers of services.
%From Fig. \ref{LPtimeratio}, it can be clearly seen that it is much more efficient to solve relaxation \eqref{lp} than relaxation \eqref{milplp}.
%%
%In all cases, solving relaxation \eqref{lp} is 5+ times faster than solving relaxation \eqref{milplp}.
%%In some cases ($|\mathcal{K}|\geq 5$), solving relaxation \eqref{lp} is even 10+ times faster than solving relaxation \eqref{milplp}.
%%
%In addition, we can observe from Fig. \ref{LPtimeratio} that, the {CPU time ratio} generally increases with the number of services.
%This is mainly due to the fact that the numbers of variables and constraints in relaxation \eqref{milplp} increase much faster than those in relaxation \eqref{lp} as the number of services increases.

Next,  we compare the optimal values of the two LP relaxations \eqref{milplp} and \eqref{lp}.
As it has been shown in Corollary \ref{strongerLP}, the LP relaxation \eqref{lp} is stronger than the natural LP relaxation \eqref{milplp}, and in  Example \ref{example1}, relaxation \eqref{lp} returns a solution with a much larger communication delay, as compared to that returned by relaxation \eqref{milplp}.
Therefore, we compare the total communication delay of the solutions returned by the two relaxations.
To do this, we need to eliminate the effect of the number of activated nodes and the total NFV delay in the objective function.
This can be done by fixing variables $\boldsymbol{x}$ and $\boldsymbol{y}$, in relaxations \eqref{milplp} and \eqref{lp},  to the solution returned by solving problem \eqref{newmip}.
We compare the relative gap, which is defined as follows:
$$
	\text{Gap} = \frac{\text{D(MILP)} - \text{D(LP-II)}}{\text{D(MILP)} - \text{D(LP-I)}},
$$
where D(MILP), D(LP-I), and  D(LP-II) are the total communication delay returned by solving problems \eqref{newmip}, \eqref{milplp}, and \eqref{lp}, respectively.
The relative gap reflects the tighterness of relaxation \eqref{lp} over that of relaxation \eqref{milplp}, i.e., the smaller the relative gap, the stronger relaxation \eqref{lp} (as compared with relaxation \eqref{milplp}).
\rev{We remark that the {relative gap} is a widely used performance measure in the integer programming community \cite{Vielma2010,Fukasawa2011} to show the tightness of an LP relaxation over another LP relaxation.}

\begin{figure}[t]
	\centering
	\includegraphics[height=\figuresize]{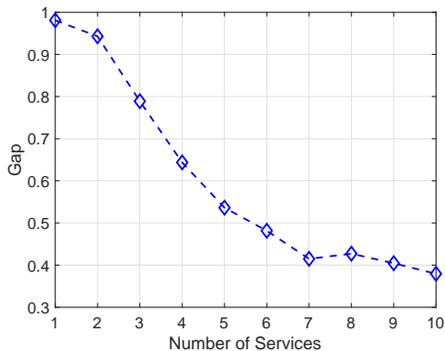}
	\caption{{Average relative gap of the total communication delays.}}
	\label{relgap}
\end{figure}

Fig. \ref{relgap} plots the average relative gap versus different numbers of services.
As observed in Fig. \ref{relgap}, in all cases, the relative gap is smaller than 1.0, showing that the LP relaxation \eqref{lp} is indeed stronger than the LP relaxation \eqref{milplp}.
In addition, with the increasing number of services, the relative gap becomes smaller.
This can be explained as follows.
As the number of services increases, the traffic in the network becomes heavier.
This further results in the situation that the traffic flows between the two nodes hosting two adjacent functions are likely to transmit over multiple paths.
Consequently, the total communication delay returned by solving LP relaxation \eqref{milplp} becomes smaller and relaxation \eqref{milplp} becomes looser.

%This clearly shows that relaxation \eqref{lp} is indeed much stronger than relaxation \eqref{milplp}, especially when the number of services is large.

Based on the above computational results, we can conclude that the proposed relaxation \eqref{lp} significantly outperforms the natural relaxation \eqref{milplp} in terms of the solution efficiency and providing a better LP bound.
As it will be seen in the next subsection, the solution efficiency and the tightness of the proposed LP relaxation \eqref{lp} plays a crucial role in the effectiveness and efficiency of the proposed \LPRR algorithm.

\subsection{Comparison of Proposed Algorithm and Those in \cite{Chen2020}, \cite{Chowdhury2012}, \cite{Yu2008}, and \revv{\cite{Zhang2012}}}
\label{compofalgs}
%\rev{In this subsection, we compare the performance of the proposed LPRR algorithm with the exact approach using standard MILP solvers (called EXACT) in \cite{Chen2020}, the GREEDY algorithm (which is adapted from the one in \cite{Yu2008}), and the LP rounding (LPR) algorithm (which is adapted from the one in \cite{Chowdhury2012}).}
\revv{In this subsection, we compare the performance of the proposed \LPRR algorithm with the exact approach using standard MILP solvers (called \EXACT) in \cite{Chen2020}, the \LPR algorithm in \cite{Chowdhury2012},  the \GREEDY algorithm in \cite{Yu2008}, and the \PSO algorithm in \cite{Zhang2012}.}
%\rev{GREEDY maps VNFs into cloud nodes using a greedy heuristic while $\LPR$ first solves the LP relaxation \eqref{lp} of the network slicing problem to obtain its solution and then uses a rounding strategy to map VNFs to cloud nodes based on this solution.}
%\rev{The two algorithms then find paths connecting two adjacent VNFs in the network by solving a multicommodity flow problem \cite[Chapter 17]{Ahuja1993}, which can be solved in polynomial time.}
%\rev{Due to the space limit, we provide the detailed description of the GREEDY and LPR in Section II of \cite{Chen2021d}.} 
To address the advantage of using the proposed LP relaxation \eqref{lp} in the \LPRR algorithm, we compare the \LPRR algorithm with \LPRRS, which embeds the natural LP relaxation \eqref{milplp} into Algorithm \ref{roundingx} while still embeds the proposed LP relaxation \eqref{lp} into Algorithm \ref{al:routing}.
It is worthwhile remarking that the natural LP relaxation \eqref{milplp} cannot be used in Algorithm \ref{al:routing} as we generally cannot find a traffic routing strategy based on the  solution of the natural LP relaxation \eqref{milplp}; see the example in Appendix \ref{detailsofex1}.
\revv{In addition, to evaluate the effect of the LP dynamic rounding strategy for solving the VNF placement subproblem in Algorithm \ref{roundingx}, we also compare the proposed \LPRR algorithm with \LPRRSS in which an LP static rounding strategy is used; see the last paragraph of Section \ref{comexisting} for a detailed discussion of \LPRRSS.} 
\begin{figure}[t]
	\centering
	\includegraphics[height=\figuresize]{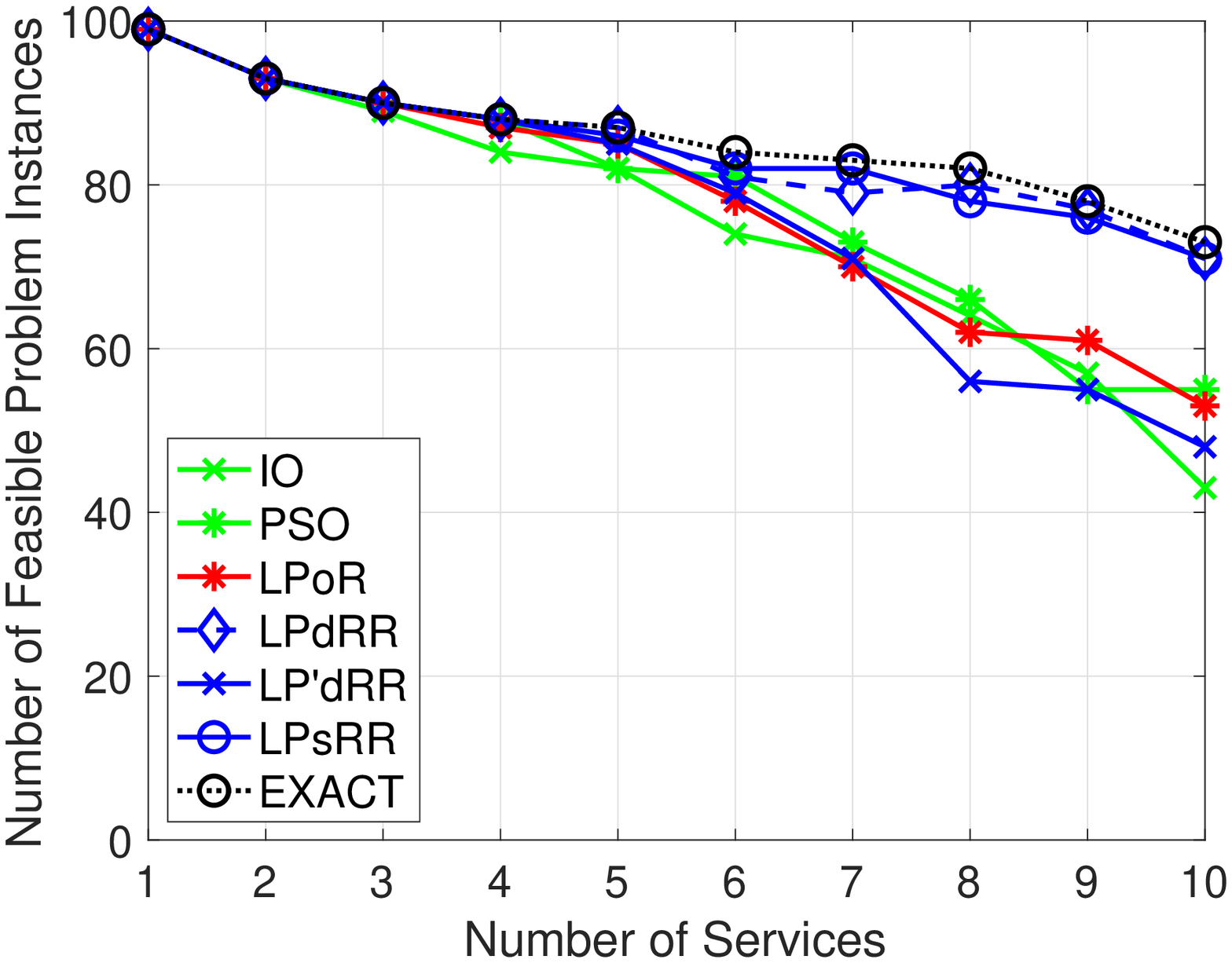}
	\caption{\rev{The number of feasible problem instances solved by \revv{\GREEDY, \PSO, \LPR, \LPRR, \LPRRS, \LPRRSS, and \EXACT}.}}
	\label{nfeas}
\end{figure}
\begin{figure}[t]
	\centering
	\includegraphics[height=\figuresize]{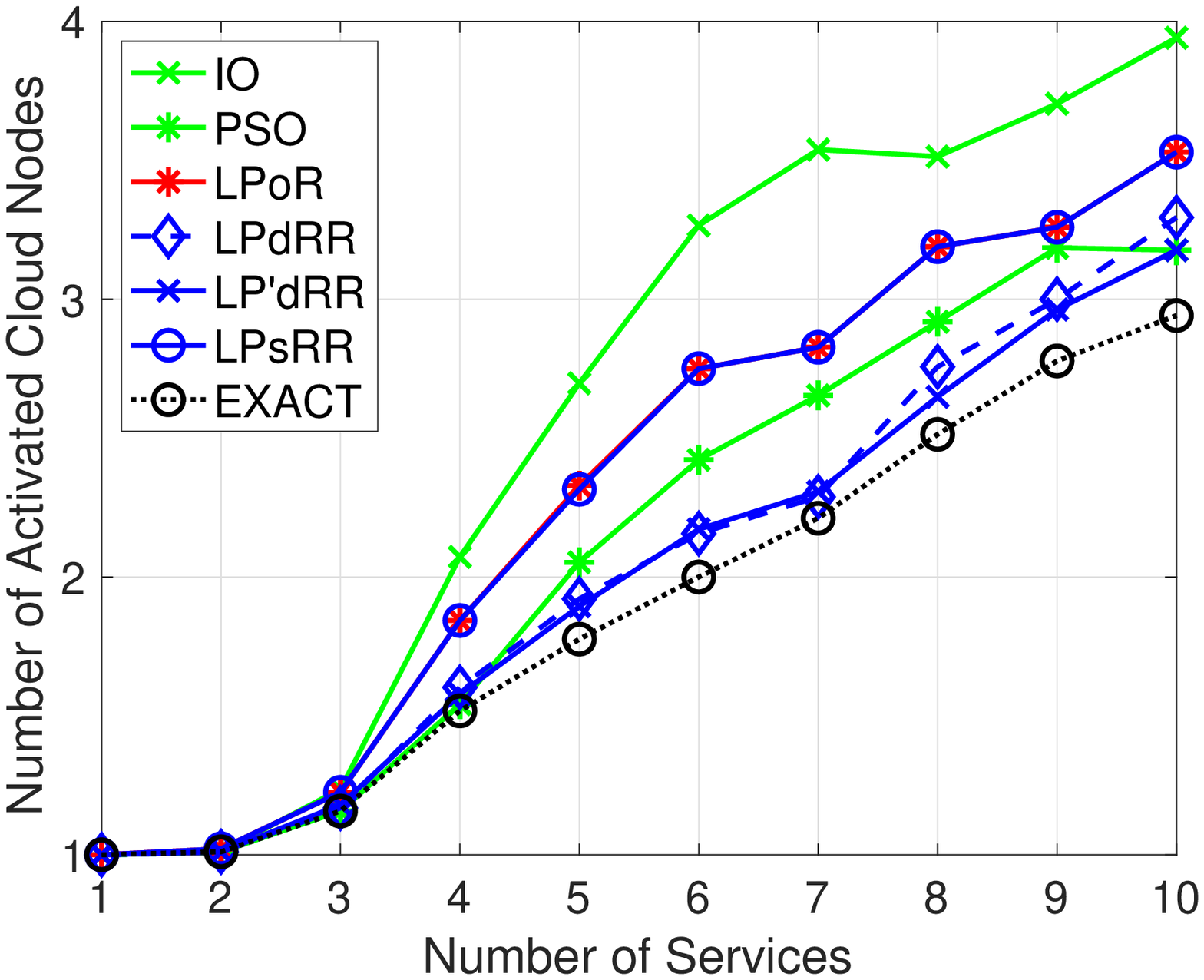}
	\caption{\rev{Average number of activated cloud nodes of  \revv{\GREEDY, \PSO, \LPR, \LPRR, \LPRRS, \LPRRSS, and \EXACT}.}}
	\label{sumy}
\end{figure}
\begin{figure}[t]
	\centering
	\includegraphics[height=\figuresize]{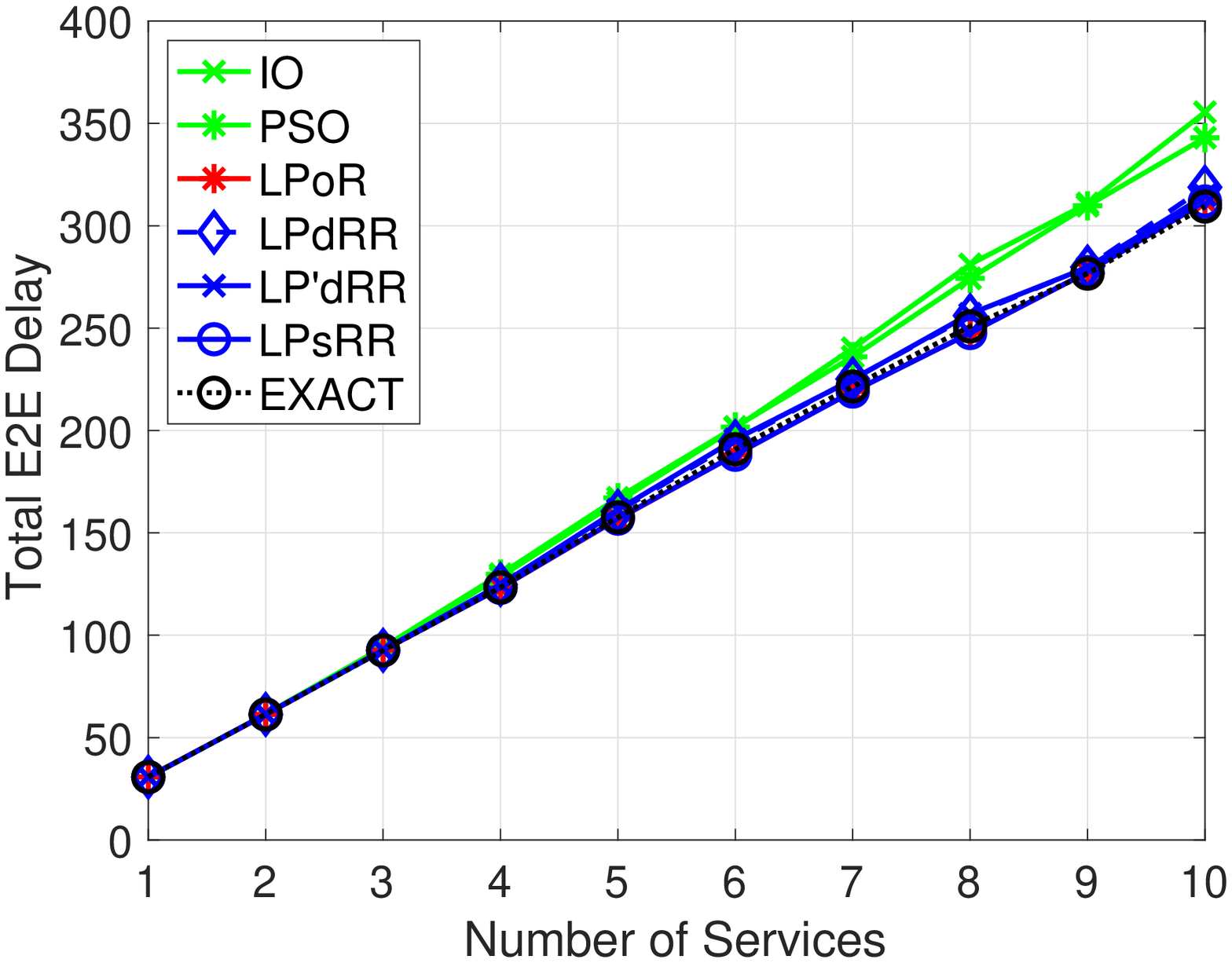}
	\caption{\revv{Total E2E delay solved by \GREEDY, \PSO, \LPR, \LPRR, \LPRRS, \LPRRSS, and \EXACT.}}
	\label{delay}
\end{figure}
%\vspace{0.1cm}\\{\bf\noindent$\bullet$ Effectiveness of $\LPRR$\vspace{0.1cm}}
\subsubsection{Effectiveness of $\LPRR$}
\revv{
	Figs. \ref{nfeas}, \ref{sumy}, and \ref{delay} plot the number of feasible problem instances, the average number of activated nodes, and the total E2E delay solved by \GREEDY, \PSO, \LPR, \LPRR, \LPRRS, \LPRRSS, and \EXACT, respectively.
}
First, we can observe from Fig. \ref{nfeas} that compared with $\LPRRS$, $\LPRR$ can solve a larger number of problem instances than $\LPRRS$, especially when the number of services is large.
This clearly shows that the advantage of using the proposed LP relaxation \eqref{lp} in the proposed \LPRR algorithm, i.e., it enables the \LPRR algorithm to find feasible solutions for much more problem instances.
\revv{Second, \LPRR and \LPRRSS are comparable in terms of finding feasible solutions but \LPRR can return a solution with a smaller number of activated cloud nodes than \LPRRSS, as illustrated in Figs. \ref{nfeas} and \ref{sumy}.
This highlights the advantage of the dynamic strategy of rounding the LP solution in the proposed \LPRR algorithm.}
\revv{Finally, we can see the effectiveness of the proposed algorithm \LPRR~over \GREEDY, \PSO, and \LPR in Figs. \ref{nfeas}-\ref{delay}.}
In particular, as shown in Fig.~\ref{nfeas}, using the proposed algorithm \LPRR, we can find feasible solutions for much more problem instances, compared with using \revv{\GREEDY, \PSO, and \LPR}.
Indeed, \LPRR~finds feasible solutions for almost all (truly) feasible problem instances (as EXACT is able to find feasible solutions for all (truly) feasible problem instances and the difference of the number of feasible problem instances solved by EXACT and \LPRR~is small in Fig.~\ref{nfeas}).
\revv{From Fig. \ref{sumy}, we can observe that using \LPRR, the number of activated cloud nodes is {generally} smaller  than that of using \GREEDY, \PSO, and \LPR, respectively.
From Fig. \ref{delay}, we can observe  that (i) the E2E delays returned by \GREEDY and \PSO are fairly larger than those returned by \LPR, \LPRR, \LPRRS, \LPRRSS, and \EXACT; and
	(ii) the E2E delays returned by \LPR, \LPRR, \LPRRS, \LPRRSS, and \EXACT are comparable.
	The reason for this is that \LPR, \LPRR, \LPRRS, \LPRRSS, and \EXACT explicitly consider the E2E delay when solving the VNF placement subproblem while \GREEDY and \PSO do not explicitly consider the E2E delay when solving the VNF placement subproblem (and hence return a solution with a higher E2E delay).}
\begin{figure}[t]
	\centering
	\includegraphics[height=\figuresize]{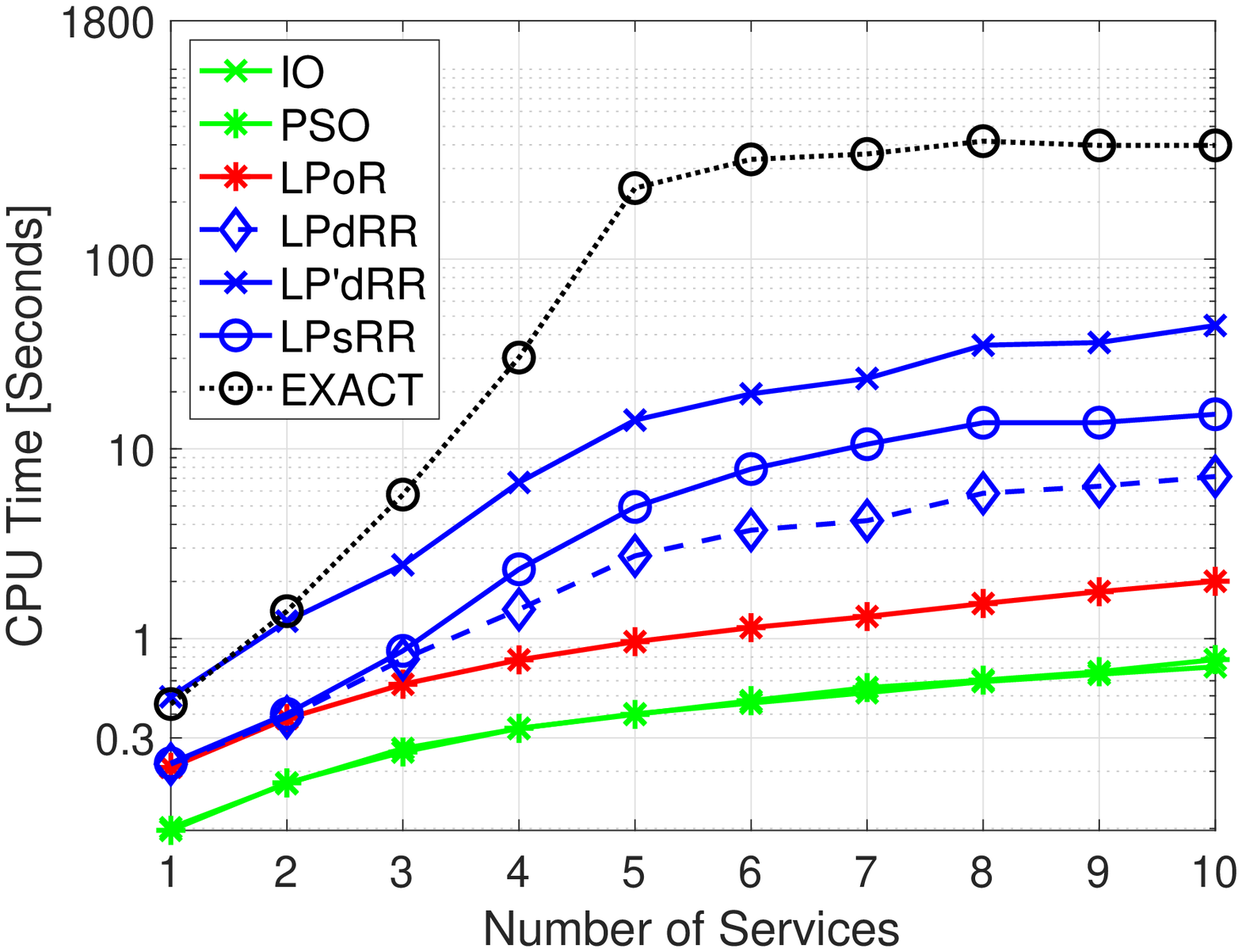}
	\caption{Average CPU time of \revv{\GREEDY, \PSO, \LPR, \LPRR, \LPRRS, \LPRRSS, and \EXACT.}}
	\label{time}
\end{figure}
%{\bf\noindent$\bullet$ Efficiency of $\LPRR$ \vspace{0.1cm}}\\
\subsubsection{Efficiency of $\LPRR$}
\revv{The comparison of the solution time of \GREEDY, \PSO, \LPR, \LPRR, \LPRRS, \LPRRSS, and \EXACT is plot in Fig. \ref{time}.
First, as expected, \LPRR is much more computationally efficient than \LPRRS and \LPRRSS, especially when the number of services is large.
The efficiency of \LPRR over \LPRRS mainly comes from the fact that solving the proposed LP relaxation is much faster than solving the natural LP relaxation, as demonstrated in Fig. \ref{LPtime}, while the efficiency of \LPRR over \LPRRSS is due to the fact that the number of solved LPs in \LPRR is much smaller than that solved in \LPRRSS, as will be illustrated in Fig. \ref{niter}.}
Second, compared with \EXACT, \LPRR~is significantly more efficient and the solution efficiency of \GREEDY, \revv{\PSO}, \LPR, and \LPRR~is comparable.
Overall, \LPR and \LPRR are several times slower than \revv{\GREEDY and \PSO}. 
%For example, when $|\mathcal{K}|=10$, the average CPU time of \revv{\GREEDY and \PSO} are \revv{$0.71$ seconds and $0.78$ seconds}, while the average CPU times of \LPR and \LPRR are \revv{$2.00$ seconds and $7.16$ seconds}, respectively. 
%\revv{In sharp contrast, the average CPU time is $397.27$ seconds}.

%Here we scale the solution time of LPR to be 1 and compute the CPU time ratio as follows:
%\begin{equation*}
%	{\text{CPU time ratio}} = \frac{\text{T(\LPRR)~(T(\LPRRS), or T(EXACT))}}{\text{T(LPR)}},
%	%	\text{ratio} = \frac{\text{T}\eqref{mip}}{\text{T}\eqref{newmip}},
%\end{equation*}
%where $\text{T(LPR)}$, $\text{T(\LPRR)}$, $\text{T(\LPRRS)}$, and $\text{T(EXACT)}$ are the CPU time taken by LPR, $\LPRR$, $\LPRRS$, and EXACT, respectively.
%%
%First, as expected, \LPRR~is much more computationally efficient than $\LPRRS$, especially when the number of services is large.
%This is mainly due to the fact that solving the proposed LP relaxation is much faster than solving the natural LP relaxation, as demonstrated in Fig. \ref{LPtime}.
%Second, compared with EXACT, \LPRR~is significantly more efficient and the solution efficiency of \LPRR~and LPR is comparable.
%
%Indeed, \LPRR~is at most five times slower than LPR in all cases while EXACT is even 100+ times slower than LPR when the problem is large (i.e., $|\mathcal{K}| \geq 5$).
%

\begin{figure}[t]
	\centering
	\includegraphics[height=\figuresize]{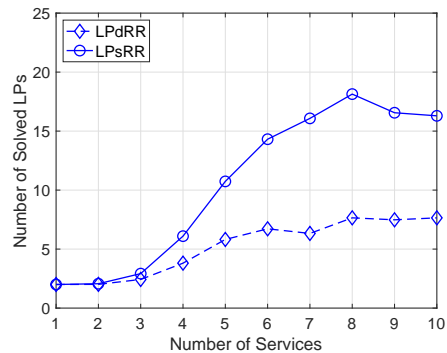}
	\caption{{Average number of solved LPs in \revv{\LPRR and \LPRRSS}.}}
	\label{niter}
\end{figure}

To gain more insight into the computational efficiency of the proposed algorithms \revv{\LPRR and \LPRRSS}, we plot the number of solved LPs in these two algorithms in Fig. \ref{niter}.
From this figure, we observe that the number of solved LPs in \LPRR increases with the number of services, \revv{but is much smaller than its theoretical upper bound in \eqref{upperLP}}.
For instance, when $|\mathcal{K}|= 10$, the average number of solved LPs in $\LPRR$ is \revv{$7.6$} while its upper bound for is $6*(10*3)+10=190$.
This clearly shows that the proposed $\LPRR$ algorithm works well in practice.
In addition, compared with \revv{\LPRRSS}, \LPRR solves a smaller number of LPs, which is the main reason why \LPRR is faster than \revv{\LPRRSS}.

\subsection{Summary of Simulation Findings}
From the above simulation results, we can conclude that:
\begin{itemize}
	\item [$\bullet$]
	      Compared with the natural LP relaxation \eqref{milplp}, our proposed relaxation \eqref{lp} is much stronger and more compact, which plays a crucial role in the effectiveness and efficiency of our proposed \LPRR algorithm.
	\item [$\bullet$] %Our simulation results illustrate the effectiveness and efficiency of our proposed algorithm \LPRR. 
	      When comparing with the existing state-of-the-art algorithms in \cite{Chen2020}, \cite{Chowdhury2012}, \cite{Yu2008}, and \cite{Zhang2012}, our proposed $\LPRR$ achieves a better trade-off between the solution quality and the computational efficiency.
	      More specifically, compared with \revv{\GREEDY in \cite{Yu2008}, \PSO  in \cite{Zhang2012}, and \LPR} in \cite{Chowdhury2012}, it is able to find a much better solution; compared with EXACT in \cite{Chen2020}, it is much more computationally efficient.
\end{itemize}

\section{Conclusions}
\label{sec:conclusion}

In this paper, we have proposed an efficient algorithm called \LPRR for solving the network slicing problem.
The proposed algorithm is a two-stage LP based algorithm that places the virtual network functions of all services into cloud nodes using an LP \revv{dynamic} rounding procedure in the first stage and finds the traffic routing strategies of all services using an iterative LP refinement procedure in the second stage.
Three key features of the proposed algorithm, which make it particularly suitable to solve the large-scale network slicing problems with a high-quality solution, are: (i) it is based on a newly proposed LP relaxation \eqref{lp}, which, compared with the natural LP relaxation \eqref{milplp}, is much stronger (in terms of providing a better LP bound) and more compact (in terms of smaller numbers of variables and constraints); and (ii) it takes the global information of the problem (i.e., the traffic routing of all services) into consideration in the first stage (in contrast to \cite{Chowdhury2012} which only takes the local information of the problem); (iii) the worst-case complexity of the proposed algorithm is polynomial.
Simulation results show that our proposed algorithm is able to achieve a better trade-off between the solution quality and the computational efficiency than the existing state-of-the-art algorithms.
%In particular, when compared with the one in \cite{Chowdhury2012}, our proposed algorithm is able to find a much better solution; when compared with the one in \cite{Chen2020}, our proposed algorithm is much more computationally efficient.

%\input{section_appendix}
%\input{section_appendixcomplexity}
\appendices
\section{Details of Example \ref{example1}}
\label{detailsofex1}

Here we will show that solving the nonlinear relaxation \eqref{nlp} will return a solution with the communication delay being $1.5$, 
solving the natural LP relaxation \eqref{milplp} will return a solution with the communication delay being $1.25$, 
\revvv{and solving the novel LP relaxation \eqref{lp} will return a solution with the communication delay being $1.5$}.

To distinguish the two links in Fig. \ref{example}, we denote the top one with communication delay $1$ and the bottom one with communication delay $2$ as $a$ and $b$, respectively.
Let $\theta$ denote the delay incurred by the traffic flow from node $S$ to node $D$.
For $p=1,2$, let $r(p)$ denote the data rate on the $p$-th path, $z_a(p)$ and $z_{b}(p)$ denote whether or not links $a$ and $b$ are used by the $p$-th path, respectively, and $r_a(p)$ and $r_{b}(p)$ denote the associated data rate on links $a$ and $b$, respectively.
The nonlinear relaxation \eqref{nlp} for this example reduces to
\begin{equation}
	\begin{aligned}
	\theta_{\text{NLP}} := & \min_{\boldsymbol{r},\,\boldsymbol{z},\,{\theta}} ~\theta                                                         \\
	                       & {~~\text{s.t.}}~~                                                             \text{\eqref{exflowcons1}--\eqref{exvar}}, ~\text{where} \label{exnlp}
	 \end{aligned}
\end{equation}
\vspace*{-0.6cm}
\begin{multicols}{2}
	\noindent
	\begin{align}
		  & z_{a}(1)+z_{b}(1)=1, \label{exflowcons1}            \\
		  & z_{a}(2)+z_{b}(2)=1,\label{exflowcons2}             \\
		  & r_{a}(1) + r_{a}(2) \leq 0.5,\label{excap1}         \\
		  & r_{b}(1) + r_{b}(2) \leq 0.5,\label{excap2}\\ 
		  & r(1)+r(2) = 1,\label{extotal}         
	\end{align}
	\begin{align}
		  & r_{a}(1)=r(1)z_{a}(1), \label{exrrz1} \\
		  & r_{b}(1)=r(1)z_{b}(1), \label{exrrz2} \\
		  & r_{a}(2)=r(2)z_{a}(2), \label{exrrz3} \\
		  & r_{b}(2)=r(2)z_{b}(2),\label{exrrz4}
	\end{align}
\end{multicols}
\vspace{-1.05cm}
\begin{align}
	& \revv{\theta =\max\{ z_{a}(1) + 2z_{b}(1),z_{a}(2) + 2z_{b}(2)\}},\label{nlpdelay}\\
	& 0\leq r(p), r_a(p), r_{b}(p), z_{a}(p), z_{b}(p)\leq 1, ~p \in \{1,2\}.	\label{exvar}
\end{align}
\noindent The natural LP relaxation \eqref{milplp} for this example reduces to
\begin{equation}
\begin{aligned}
	\theta_{\text{LP-I}} := & \min_{\boldsymbol{r},\,\boldsymbol{z},\,{\theta}} ~  \theta                                                                                                            \\
	                        & {~~\text{s.t.}} ~~ \text{\eqref{exflowcons1}--\eqref{extotal}}, \text{\eqref{exvar},~\eqref{exlin11}--\eqref{exdelay2}}, ~\text{where}~ \label{exmilplp}  
\end{aligned}	       
\end{equation}
\vspace*{-0.6cm}
\begin{figure*}
	\begin{multicols}{3}
		\noindent
		\begin{align}
			r_a(1) & \leq r(1), \label{exlin11}%\tag{25a} 
			\\
			r_b(1) & \leq r(1), \label{exlin21}%\tag{26a} 
			\\
			r_a(2) & \leq r(2), \label{exlin31}%\tag{27a} 
			\\
			r_b(2) & \leq r(2), \label{exlin41}%\tag{28a}
		\end{align}
		\begin{align}
			r_a(1) & \leq z_{a}(1),  \label{exlin12}%\tag{25b} 
			\\
			r_b(1) & \leq z_{b}(1),  \label{exlin22}%\tag{26b} 
			\\
			r_a(2) & \leq z_{a}(2),  \label{exlin32}%\tag{27b} 
			\\
			r_b(2) & \leq z_{b}(2),  \label{exlin42}%\tag{28b}
		\end{align}
		\begin{align}
			r_a(1) & \geq r(1)+z_a(1)-1,\label{exlin13}%\tag{25c}  
			\\
			r_b(1) & \geq r(1)+z_b(1)-1,\label{exlin23}%\tag{26c}  
			\\
			r_a(2) & \geq r(2)+z_a(2)-1,\label{exlin33}%\tag{27c}  
			\\
			r_b(2) & \geq r(2)+z_b(2)-1.\label{exlin43}%\tag{28c}
		\end{align}
	\end{multicols}
\vspace*{-5mm}
\hrulefill
\vspace*{-5mm}
\end{figure*}

\begin{multicols}{2}
	\noindent
	\begin{align}
		& \revv{\theta \geq z_{a}(1) + 2z_{b}(1)}, \label{exdelay1}     
	\end{align}
	\begin{align}
		& \revv{\theta \geq z_{a}(2) + 2z_{b}(2)}. \label{exdelay2}
	\end{align}
\end{multicols}

\revv{
\noindent The novel LP relaxation \eqref{lp} reduces to 
\begin{equation}
	\begin{aligned}
		\theta_{\text{LP-II}} := & \min_{\boldsymbol{r},\,{\theta}} ~\theta                                                         \\
		& {~~\text{s.t.}}~~                                                             
		\text{\eqref{exdelay1-1}--\eqref{excap2-1}}, ~\text{where} \label{exlp2}
	\end{aligned}
\end{equation}
\begin{multicols}{2}
	\noindent
	\begin{align}
		& \theta = r_{a}(1) + 2r_{b}(1), \label{exdelay1-1}  \\
		& r_{a}(1)+r_{b}(1)=1, \label{exflowcons1-1}            
	\end{align}
	\begin{align}
		& r_{a}(1)  \leq 0.5,\label{excap1-1}         \\
		& r_{b}(1)  \leq 0.5,\label{excap2-1}
	\end{align}
\end{multicols}
Below we shall show that $ \theta_{\text{NLP}} = 1.5 $,
$ \theta_{\text{LP-I}} = 1.25 $, and $ \theta_{\text{LP-II}} = 1.5 $, separately.}\\
{\noindent$\bullet$~Proof of $\theta_{\text{NLP}} = 1.5$.\vspace{0.1cm}\\}
First, we give a point $(\bar{\boldsymbol{r}},\bar{\boldsymbol{z}},\bar{{\theta}})$ as follows:
\begin{equation*}
	\begin{aligned}
		  &   &   & \bar{\theta}=1.5, ~ \bar{r}(1)=1, ~ \bar{r}(2) = 0,                        \\
		  &   &   & \bar{z}_{a}(1) =\bar{z}_{b}(1) =0.5, ~	\bar{r}_{a}(1) =\bar{r}_{b}(1) =0.5, \\
		  &   &   & \bar{z}_{a}(2) =\bar{z}_{b}(2) =0.5, ~	\bar{r}_{a}(2) =\bar{r}_{b}(2) =0.
	\end{aligned}
\end{equation*}
Obviously, $(\bar{\boldsymbol{r}},\bar{\boldsymbol{z}},\bar{{\theta}})$ is feasible for problem \eqref{exnlp}, and hence we have $\theta_{\text{NLP}} \leq 1.5$.
It remains to show $\theta_{\text{NLP}} \geq 1.5$.
Indeed, summing up constraints \eqref{exrrz1}--\eqref{exrrz4} and using constraints \eqref{exflowcons1}, \eqref{exflowcons2}, and \eqref{extotal}, we have
\begin{equation*}
	r_{a}(1)+r_{b}(1)+r_{a}(2)+r_{b}(2)\qquad\qquad\qquad\qquad\qquad
\end{equation*}
\begin{equation*}
	\begin{aligned}
		  & =r(1)(z_{a}(1)+z_{b}(1))+r(2)(z_{a}(2)+z_{b}(2)) \\
		  & =r(1)+r(2)=1.
	\end{aligned}
\end{equation*}
This implies that constraints \eqref{excap1} and \eqref{excap2} must hold with equalities:
\begin{align}
	r_{a}(1) + r_{a}(2) & = 0.5,\nonumber%\label{excap11}
	\\ 
	r_{b}(1) + r_{b}(2) & = 0.5.\label{excap21}
\end{align}
Substituting \eqref{exrrz2} and \eqref{exrrz4} into  \eqref{excap21}, we have
\begin{equation*}
	r_{b}(1) + r_{b}(2)= r(1)z_{b}(1)+ r(2)z_{b}(2) = 0.5, 
\end{equation*}
which, together with constraint \eqref{extotal}, immediately shows $\max\{ z_b(1), z_{b}(2) \} \geq 0.5$.
Then, it follows from constraints \revv{\text{\eqref{exflowcons1}--\eqref{exflowcons2}} and \eqref{nlpdelay}} that 
\begin{equation*}
	\begin{aligned}
		\theta_{\text{NLP}} & \revv{=} \max\{ z_{a}(1)+ 2z_{b}(1), z_{a}(2)+ 2z_{b}(2) \} \\
		                    & =   \max\{ 1+z_{b}(1), 1+z_{b}(2) \}\geq 1.5.
	\end{aligned}
\end{equation*}
{\noindent$\bullet$~Proof of $\theta_{\text{LP-I}} = 1.25                                                                                                                                                                                                                 $.\vspace{0.1cm}\\}
Let  $(\hat{\boldsymbol{r}},\hat{\boldsymbol{z}},\hat{{\theta}})$ be given as follows:
\begin{equation*}
	\begin{aligned}
		  &   &   & \hat{\theta}=1.25, ~ \hat{r}(1)=0.5, ~ \hat{r}(2) = 0.5,                               \\
		  &   &   & \hat{z}_{a}(1) =0.75,~\hat{z}_{b}(1) =0.25, ~	\hat{r}_{a}(1) =0.25,~\hat{r}_{b}(1)=0,   \\
		  &   &   & \hat{z}_{a}(2) = 0.75,~\hat{z}_{b}(2) =0.25, ~	\hat{r}_{a}(2) =0.25,~\hat{r}_{b}(2) =0.
	\end{aligned}
\end{equation*}
It is simple to check that $(\hat{\boldsymbol{r}},\hat{\boldsymbol{z}},\hat{{\theta}})$ is feasible for problem \eqref{exmilplp} (though we cannot find a traffic routing strategy based on this solution).
As a result, we have $\theta_{\text{LP-I}} \leq 1.25$.
To seek a solution with a better objective value for problem \eqref{exmilplp}, we may assume that
\begin{equation*}
	\begin{aligned}
		z_{a}(1) + 2z_{b}(1) & = 1+ z_{b}(1) \leq 1.25,  \\
		z_{a}(2) + 2z_{b}(2) & = 1+ z_{b}(2) \leq 1.25.
	\end{aligned}
\end{equation*}
Notice that the above equalities follow from \eqref{exflowcons1} and \eqref{exflowcons2}.
Then, we have $z_{b}(1) \leq 0.25$ and $z_{b}(2) \leq 0.25$, and it follows from constraints \eqref{exflowcons1} and \eqref{exflowcons2} that $z_{a}(1)\geq 0.75$ and $ z_{a}(2) \geq 0.75$.
Summing up constraints \eqref{exlin13} and \eqref{exlin33} yields
\begin{equation*}
	\begin{aligned}
		r_a(1) + r_a(2) & \geq  r(1)+ r(2)+ z_{a}(1)+ z_{a}(2) - 2                    \\
		                & = 1+ z_{a}(1)+ z_{a}(2) - 2 ~ (\text{from}~\eqref{extotal}) \\
		                & = z_{a}(1)+ z_{a}(2) -1                                     \\
		                & \geq 0.5.~ (\text{from}~z_{a}(1)\geq 0.75~\text{and}~z_{a}(2)	\geq 0.75)
	\end{aligned}
\end{equation*}
By constraint \eqref{excap1}, the above inequality must hold with equality, and hence we have $z_{a}(1)= z_{a}(2) = 0.75$, or equivalently $z_{b}(1) =  z_{b}(2) = 0.25$.
As a result, $\theta_{\text{LP-I}} = 1.25$.\\[5pt]
\revv{
	{\noindent$\bullet$~Proof of $\theta_{\text{LP-II}} = 1.5                                                                                                                                                                                                                 $.\vspace{0.1cm}\\}
	It is easy to see  that point $(\tilde{\boldsymbol{r}},\tilde{{\theta}})$ defined by
	\begin{equation*}
		\begin{aligned}
			\tilde{\theta}=1.5, ~\tilde{r}_a(1) =\tilde{r}_{b}(1)=0.5
		\end{aligned}
	\end{equation*}
	is an optimal solution of relaxation \eqref{exlp2}, implying  $\theta_{\text{LP-II}} = 1.5$.\\[5pt]
}
\rev{{\bf Acknowledgments} We would like to thank the \revv{three} anonymous reviewers for their insightful comments.}


\begin{thebibliography}{1}
	%	\bibitem{Chen2020}
	%	W.-K. Chen, Y.-F. Liu, A. De Domenico, and Z.-Q. Luo, ``Network slicing for service-oriented networks with flexible routing and guaranteed E2E latency,'' \emph{Proceedings of 21st IEEE International Workshop on Signal Processing Advances in Wireless Communications (to appear)}, 2020. [Online]. Available: https://arxiv.org/pdf/2002.07380.pdf

	\bibitem{Chen2021}
	W.-K. Chen, Y.-F. Liu, Y.-H. Dai, and Z.-Q. Luo, ``An efficient linear programming rounding-and-refinement algorithm for large-scale network slicing problem,'' in \emph{Proceedings of 46th IEEE International Conference on Acoustics, Speech and Signal Processing (ICASSP)}, Toronto, Canada, June 2021, pp. 4735-4739.

	\bibitem{Mijumbi2016}
	R. Mijumbi, J. Serrat, J.-L. Gorricho, N. Bouten, F. De Turck, and R. Boutaba, ``Network function virtualization: State-of-the-art and research challenges,'' \emph{IEEE Communications Surveys \& Tutorials}, vol. 18, no. 1, pp. 236-262, Firstquarter 2016.


	%		\bibitem{Oljira2017}
	%		D. B. Oljira, K. Grinnemo, J. Taheri, and A. Brunstrom, ``A model for QoS-aware VNF placement and provisioning,'' in  \emph{Proceedings of IEEE Conference on Network Function Virtualization and Software Defined Networks (NFV-SDN)}, 2017, pp. 1-7.
	%		

	\bibitem{Chen2020}
	W.-K. Chen, Y.-F. Liu, A. De Domenico, Z.-Q. Luo, and Y.-H. Dai. ``Optimal network slicing for service-oriented networks with flexible routing and guaranteed E2E latency,'' \emph{IEEE Transactions on Network and Service Management}, vol. 18, no. 4, pp. 4337-4352, December 2021.


%	\bibitem{Domenico2019}
%	A. De Domenico, Y.-F. Liu, and W. Yu, ``Optimal computational resource allocation and network slicing deployment in 5G hybrid C-RAN,'' in \emph{Proceedings of IEEE International Conference on Communications (ICC)}, Shanghai, China, May 2019, pp. 1-6.
	\bibitem{Domenico2020}
	A. De Domenico, Y.-F. Liu, and W. Yu, ``Optimal virtual network function deployment for 5G network slicing in a hybrid cloud infrastructure,'' \emph{IEEE Transactions on Wireless Communications}, vol. 19, no. 12, pp. 7942-7956, December 2020.

	\bibitem{Jarray2012}
	A. Jarray and A. Karmouch, ``Periodical auctioning for QoS aware virtual network embedding,'' in \emph{Proceedings of IEEE 20th International Workshop on Quality of Service (IWQoS)}, Coimbra, Portugal, June 2012, pp. 1-4.

	\bibitem{Addis2015}
	B. Addis, D. Belabed, M. Bouet, and S. Secci, ``Virtual network functions placement and routing optimization,'' in \emph{Proceedings of IEEE 4th International Conference on Cloud Networking (CloudNet)}, Niagara Falls, Canada, October 2015, pp. 171-177.

	\bibitem{Ayoubi2019}
	S. Ayoubi, S. Sebbah, and C. Assi, ``A logic-based Benders decomposition approach for the VNF assignment problem,'' \emph{IEEE Transactions on Cloud Computing}, vol. 7, no. 4, pp. 894-906, October-December. 2019.

	\bibitem{Hu2013}
	Q. Hu, Y. Wang, and X. Cao, ``Resolve the virtual network embedding problem: A column generation approach,'' in \emph{Proceedings IEEE INFOCOM}, Turin, Italy, April 2013, pp. 410-414.

	\bibitem{Jarray2015}
	A. Jarray and A. Karmouch, ``Decomposition approaches for virtual network embedding with one-shot node and link mapping,'' \emph{IEEE/ACM Transactions on Networking}, vol. 23, no. 3, pp. 1012-1025, June 2015.

	\bibitem{Mijumbi2015}
	R. Mijumbi, J. Serrat, J. Gorricho, and R. Boutaba, ``A path generation approach to embedding of virtual networks,'' \emph{IEEE Transactions on Network and Service Management}, vol. 12, no. 3, pp. 334-348, September 2015.

	\bibitem{Liu2017}
	J. Liu, W. Lu, F. Zhou, P. Lu, and Z. Zhu, ``On dynamic service function chain deployment and readjustment,'' \emph{IEEE Transactions on Network and Service Management}, vol. 14, no. 3, pp. 543-553, September 2017.


	\bibitem{Zhang2017}
	N. Zhang, Y.-F. Liu, H. Farmanbar, T.-H. Chang, M. Hong, and Z.-Q. Luo, ``Network slicing for service-oriented networks under resource constraints,'' \emph{IEEE Journal on Selected Areas in Communications}, vol. 35, no. 11, pp. 2512-2521, November 2017.

	\bibitem{Chowdhury2012}
	M. Chowdhury, M. R. Rahman, and R. Boutaba, ``ViNEYard: Virtual network embedding algorithms with coordinated node and link mapping,'' \emph{IEEE/ACM Transactions on Networking}, vol. 20, no. 1, pp. 206-219, February 2012.

	\bibitem{Yu2008}
	\rev{M. Yu, Y, Yi, J. Rexford, and M. Chiang, ``Rethinking virtual network embedding: Substrate support for path splitting and migration,'' \emph{ACM SIGCOMM Computer Communication Review}, vol. 38, no. 2,  pp. 17-29, April 2008.} 

	\bibitem{Lischka2009}
	J. Lischka and H. Karl, ``A virtual network mapping algorithm based on subgraph isomorphism detection,'' in \emph{Proceedings of the 1st ACM workshop on Virtualized infrastructure systems and architectures (VISA)},  Barcelona, Spain, August 2009, pp. 81-88.

	\bibitem{Woldeyohannes2018}
	Y. T. Woldeyohannes, A. Mohammadkhan, K. K. Ramakrishnan, and Y. Jiang, ``ClusPR: Balancing multiple objectives at scale for NFV resource allocation,'' \emph{IEEE Transactions on Network and Service Management}, vol. 15, no. 4, pp. 1307-1321, December 2018.

	\bibitem{Jiang2012}
	J. W. Jiang, T. Lan, S. Ha, M. Chen, and M. Chiang, ``Joint VM placement and routing for data center traffic engineering,'' in \emph{Proceedings of IEEE INFOCOM}, Orlando, USA, March 2012, pp. 2876-2880.

	\bibitem{Qu2017}
	L. Qu, C. Assi, K. Shaban, and M. J. Khabbaz, ``A reliability-aware network service chain provisioning with delay guarantees in NFV-enabled enterprise datacenter networks,'' \emph{IEEE Transactions on Network and Service Management}, vol. 14, no. 3, pp. 554-568, September 2017.

	\bibitem{Luizelli2015}
	M. C. Luizelli, L. R. Bays, L. S. Buriol, M. P. Barcellos, and L. P. Gaspary, ``Piecing together the NFV provisioning puzzle: Efficient placement and chaining of virtual network functions,'' in \emph{Proceedings of IFIP/IEEE International Symposium on Integrated Network Management (IM)}, Ottawa, Canada, May 2015, pp. 98-106.



	%			\bibitem{Narayana2013}
	%			S. Narayana, W. Jiang, J. Rexford, and M. Chiang, ``Joint server selection and routing for geo-replicated services,'' in \emph{Proceedings of IEEE/ACM 6th International Conference on Utility and Cloud Computing}, 2013, pp. 423-428.
	%			\bibitem{Xu2013}
	%			H. Xu and B. Li, ``Joint request mapping and response routing for geo-distributed cloud services,'' in \emph{Proceedings of IEEE INFOCOM}, 2013, pp. 854-862.
	%			\bibitem{Li2016}
	%			X. Li, J. B. Rao, and H. Zhang, ``Engineering machine-to-machine traffic in 5G,'' \emph{IEEE Internet of Things Journal}, vol. 3, no. 4, pp. 609-618, 2016.
	%			\bibitem{Chua2016}
	%			F. C. Chua, J. Ward, Y. Zhang, P. Sharma, and B. A. Huberman, ``Stringer: Balancing latency and resource usage in service function chain provisioning,'' \emph{IEEE Internet Computing}, vol. 20, no. 6, pp. 22-31, 2016.
	%			\bibitem{Charikar2019}
	%			M. Charikar, Y. Naamad, J. Rexford, and X. K. Zou, ``Multi-commodity flow with in-network processing,'' in \emph{Proceedings of International Symposium on Algorithmic Aspects of Cloud Computing (ALGOCLOUD)}, 2019, pp. 73-101.
	%			\bibitem{Cheng2018}
	%			Y. Cheng, L. Yang, and H. Zhu, ``Deployment of service function chain for NFV-enabled network with delay constraint,'' in \emph{Proceedings of International Conference on Electronics Technology (ICET)}, 2018, pp. 383-386.

	%			\bibitem{Kammoun2018}
	%			A. Kammoun, N. Tabbane, G. Diaz, A. Dandoush, and N. Achir, ``End-to-end efficient heuristic algorithm for 5G network slicing,'' in \emph{Proceedings of IEEE 32nd International Conference on Advanced Information Networking and Applications (AINA)}, 2018, pp. 386-392.
	%			\bibitem{Arouk2018}
	%			O. Arouk, T. Turletti, N. Nikaein, and K. Obraczka, ``Cost optimization of Cloud-RAN planning and provisioning for 5G networks,'' in \emph{Proceedings of IEEE International Conference on Communications (ICC)}, 2018, pp. 1-6.
	%			\bibitem{Zhang2015}
	%			S. Q. Zhang, Q. Zhang, H. Bannazadeh, and A. Leon-Garcia, ``Routing algorithms for network function virtualization enabled multicast topology on SDN,'' \emph{IEEE Transactions on Network and Service Management}, vol. 12, no. 4, pp. 580–594, 2015.
	%			\bibitem{Chen2019}
	%			W.-K. Chen, Y.-F. Liu, S. Wang, A. De Domenico, and Z.-Q. Luo, ``A companion technical report of ``A new mathematical formulation of network slicing for service-oriented networks'','' 2019. [Online]. Available: http://lsec.cc.ac.cn/$\sim$yafliu/report\_network\_slicing.pdf.
	%			\bibitem{Ahuja1993}
	%			R. K. Ahuja, T. L. Magnanti, and J. B. Orlin, \emph{Network Flows}. Upper Saddle River, NJ, USA: Prentice Hall, 1993.%Prentice Hall, 1993.
	%			\bibitem{Guo2019}
	%			Z. Guo, W.-K. Chen, Y.-F. Liu, Y. Xu, and Z.-L. Zhang, ``Joint switch upgrade and controller deployment in hybrid software-defined networks,'' \emph{IEEE Journal on Selected Areas in Com- munications}, vol. 37, no. 5, pp. 1012-1028, 2019.

	\bibitem{Mohammadkhan2015}
	A. Mohammadkhan, S. Ghapani, G. Liu, W. Zhang, K. K. Ramakrishnan, and T. Wood, ``Virtual function placement and traffic steering in flexible and dynamic software defined networks,'' in \emph{Proceedings of IEEE International Workshop on Local and Metropolitan Area Networks (LANMAN)}, Beijing, China, April 2015, pp. 1-6.
	
	\bibitem{Zhang2012}
	\revv{Z. Zhang, X. Cheng, S. Su, Y. Wang, K. Shuang, and Y. Luo. ``A unified enhanced particle swarm optimization-based virtual network embedding algorithm,'' \emph{International Journal of Communication Systems}, vol. 26, pp. 1054-1073, January 2012.}

	\bibitem{Li2015}
	X. Li and C. Qian, ``The virtual network function placement problem,'' in \emph{Proceedings of IEEE Conference on Computer Communications Workshops (INFOCOM WKSHPS)}, Hong Kong, China, May 2015, pp. 69-70.

	\bibitem{Abu-Lebdeh2017}
	M. Abu-Lebdeh, D. Naboulsi, R. Glitho, and C. W. Tchouati, ``On the placement of VNF managers in large-scale and distributed NFV systems,'' \emph{IEEE Transactions on Network and Service Management}, vol. 14, no. 4, pp. 875-889, December 2017.

	\bibitem{Promwongsa2020}
	{N. Promwongsa, M. Abu-Lebdeh, S. Kianpisheh,  F. Belqasmi, R. H. Glitho, H. Elbiaze, N. Crespi, and O. Alfandi, ``Ensuring reliability and low cost when using a parallel VNF processing approach to embed delay-constrained slices,'' \emph{IEEE Transactions on Network and Service Management}, vol. 17, no. 4, pp. 2226-2241, October 2020.}

	\bibitem{Gurobi}
	Gurobi Optimization, ``Gurobi optimizer reference manual,'' \revv{2022}. [Online]. Available: http://gurobi.com.
	\bibitem{Hooker2003}
	J. H. Hooker, G. Ottosson, ``Logic-based benders decomposition'',  \emph{Mathematical Programing}, vol 96, pp. 33-60, April 2003.
	\bibitem{Conforti2014}
	M. Conforti, G. Cornu\'{e}jols, and G. Zambelli, \emph{Integer Programming}. Cham, Switzerland: Springer, 2014.

	%	\bibitem{Reyhanian2020}
	%	{N. Reyhanian, H. Farmanbar, S. Mohajer, and Z.-Q. Luo, ``Joint resource allocation and routing for service function chaining with in-subnetwork processing,'' in \emph{Proceedings of IEEE International Conference on Acoustics, Speech and Signal Processing (ICASSP)}, Barcelona, Spain, May 2020, pp. 4990-4994.}

	\bibitem{Zhang2013}
	Y. Zhang, N. Beheshti, L. Beliveau, G. Lefebvre, R. Manghirmalani, R. Mishra, R. Patneyt, M. Shirazipour, R. Subrahmaniam, C. Truchan, and M. Tatipamula, ``StEERING: A software-defined networking for inline service chaining,'' in \emph{Proceedings of 21st IEEE International Conference on Network Protocols (ICNP)}, Goettingen, Germany, October 2013, pp. 1-10.

	\bibitem{Halpern2015}
	J. Halpern and C. Pignataro, ``Service function chaining (SFC) architecture,'' 2015. [Online]. Available: https://www.rfc-editor.org/rfc/pdfrfc/rfc7665.txt.pdf.

	\bibitem{Mirjalily2018}
	G. Mirjalily and Z.-Q. Luo, ``Optimal network function virtualization and service function chaining: A survey,'' \emph{Chinese Journal of Electronics}, vol. 27, no. 4, pp. 704-717, July 2018.
	\bibitem{Ahuja1993}
	\rev{R. K. Ahuja, T. L. Magnanti, and J. B. Orlin, \emph{Network Flows}. Upper Saddle River, NJ, USA: Prentice Hall, 1993.}


	%			\bibitem{Baumgartner2015}
	%			A. Baumgartner, V. S. Reddy, and T. Bauschert . ``Combined virtual mobile core network function placement and topology optimization with latency bounds,''  in \emph{Proceedings of 4th European Workshop on Software Defined Networks}, 2015, pp. 97-102.
	%			\bibitem{Qu2019}
	%			L. Qu, C. Assi, M. Khabbaz, and Y. Ye, ``Reliability-aware service function chaining with function decomposition and multipath Routing,'' \emph{IEEE Transactions on Network and Service Management (to appear)}, 2020. [Online]. Available: https://ieeexplore.ieee.org/abstract/document/8937740/

	\bibitem{Glover1975}
	F. Glover. ``Improved linear integer programming formulations of nonlinear integer problems,''
	{\emph{Management Science}}, vol. 22, no. 4, pp. 455-460, December 1975.
	
	\bibitem{Amaldi2016}
	E. Amaldi, S. Coniglio, A.M.C.A. Koster, and M. Tieves. ``On the computational complexity of the virtual network embedding problem,''
	{\emph{Electronic Notes in Discrete Mathematics}}, vol. 52, pp. 213-220, June 2016.
	
	\bibitem{Rost2020}
	M. Rost and S. Schmid, ``On the Hardness and Inapproximability of Virtual Network Embeddings,'' \emph{IEEE/ACM Transactions on Networking}, vol. 28, no. 2, pp. 791-803, April 2020.
	
	\bibitem{Paschos2018}
	G. S. Paschos, M. A. Abdullah, and S. Vassilaras, ``Network slicing with splittable flows is hard,'' in  \emph{Proceedings of IEEE 29th Annual International Symposium on Personal, Indoor and Mobile Radio Communications (PIMRC)}, Bologna, Italy, September 2018, pp. 1788-1793.
	\bibitem{Chen2021d}
	\rev{W.-K. Chen, Y.-F. Liu, F. Liu, Y.-H. Dai, and Z.-Q. Luo, ``A companion technical report on `Towards efficient large-scale network slicing: An LP \revv{dynamic} rounding-and-refinement approach','' Academy of Mathematics and Systems Science, Chinese Academy of Sciences, Beijing, China, Technical Report, 2022. [Online]. Available: \url{http://lsec.cc.ac.cn/~yafliu/Report_Network_Slicing.pdf}
	}
	
	\bibitem{Ben-Tal2001}
	A. Ben-Tal and A. Nemirovski, \emph{Lectures on Modern Convex Optimization: Analysis, Algorithms, and Engineering Applications}. Philadelphia, USA: Society for Industrial and Applied Mathematics, 2001.
	
	\bibitem{Vazirani2001}
	\rev{
		V. V. Vazirani, \emph{Approximation Algorithms}, Berlin, Germany: Springer, 2001.}
	

	\bibitem{Liu2017a}
	Q. Liu, L. Deng, H. Zeng, and M. Chen, ``On the min-max-delay problem: NP-completeness, algorithm, and integrality gap,''  in  \emph{Proceedings of IEEE Information Theory Workshop (ITW)}, Kaohsiung, Taiwan, November 2017, pp. 21-25.
	
	\bibitem{Dinur2014}
	\revv{I. Dinur and D. Steurer, ``Analytical approach to parallel repetition,'' in \emph{Proceedings of the 46th Annual ACM Symposium on Theory of Computing}, New York, USA, May 2014, pp. 624-633.}
%	\bibitem{Garey1978}
%	\revv{M. R. Garey and D. S. Johnson, ````Strong'' NP-completeness results: Motivation, examples, and implications,'' \emph{Journal of the ACM}, vol. 25, no. 3, pp. 499-508, July 1978.	}

	
	\bibitem{Fischer2013}
	\revv{A. Fischer, J. F. Botero, M. T. Beck, H. de Meer, and X. Hesselbach, ``Virtual network embedding: A survey,'' \emph{IEEE Communications Surveys $\&$ Tutorials}, vol. 15, no. 4, pp. 1888-1906, Fourth Quarter 2013.}
	

	\bibitem{Orlowski2010}
	\rev{S. Orlowski, R. Wess\"{a}ly,  M.  Pi\'{o}ro, and A. Tomaszewski, ``SNDlib 1.0--survivable network design library,'' \emph{Networks}, vol. 55, no. 3,  pp. 276-286, May 2010.}

%	\bibitem{Renegar1988}
%	J. Renegar, ``A polynomial-time algorithm, based on Newton's method, for linear programming,'' \emph{Mathematical Programming}, vol. 40, pp. 59-93, January 1988.

%	\bibitem{Garey1978}
%	M. R. Garey and D. S. Johnson, ````Strong'' NP-completeness results: Motivation, examples, and implications, '' \emph{Journal of the ACM}, vol. 25, no. 3, pp. 499-508, 1978.

%	\bibitem{Dinur2014}
%	I. Dinur and D. Steurer, ``Analytical approach to parallel repetition,'' in \emph{Proceedings of the 46th annual ACM symposium on Theory of computing}, New York, USA, May 2014, pp. 624-633.
	
	\bibitem{Mittelmann2021} \revv{H. D. Mittelmann,  ``Latest benchmark results'', 2021. [online] \url{http://plato.asu.edu/talks/informs2021.pdf}.}
	\bibitem{Vielma2010}
	\rev{J. P. Vielma, S. Ahmed, and G. Nemhauser, 
	``Mixed-integer models for nonseparable piecewise-linear optimization: Unifying framework and extensions,''
	\emph{Operations Research}, vol. 58, no. 2, pp. 303-315, March-April 2010.
	}
	\bibitem{Fukasawa2011}
	\rev{
	R. Fukasawa and M. Goycoolea, ``On the exact separation of mixed integer knapsack cuts,'' \emph{Mathematical Programming}, vol. 128, pp. 19-41, June 2011.
}






	
		





	%	\bibitem{Chen2020}
	%	W.-K. Chen, Y.-F. Liu, A. De Domenico, Y.-H. Dai, and Z.-Q. Luo, ``Optimal network slicing for service-oriented networks with flexible routing and guaranteed E2E latency,'' 2020. [Online]. Available: https://arxiv.org/pdf/2006.13019.pdf
	%	\bibitem{Mijumbi2016}
	%	R. Mijumbi, J. Serrat, J.-L. Gorricho, N. Bouten, F. De Turck, and R. Boutaba, ``Network function virtualization: State-of-the-art and research challenges,'' \emph{IEEE Communications Surveys \& Tutorials}, vol. 18, no. 1, pp. 236-262, Firstquarter 2016.
	%	\bibitem{Zhang2013}
	%	Y. Zhang, N. Beheshti, L. Beliveau, G. Lefebvre, R. Manghirmalani, R. Mishra, R. Patneyt, M. Shirazipour, R. Subrahmaniam, C. Truchan, and M. Tatipamula, ``StEERING: A software-defined networking for inline service chaining,'' in \emph{Proceedings of 21st IEEE International Conference on Network Protocols (ICNP)}, 2013, pp. 1-10.
	%	\bibitem{Halpern2015}
	%	J. Halpern and C. Pignataro, ``Service function chaining (SFC) architecture,'' RFC 7665, October 2015. [Online]. Available: https://www.rfc-editor.org/rfc/pdfrfc/rfc7665.txt.pdf
	%	\bibitem{Mirjalily2018}
	%	G. Mirjalily and Z.-Q. Luo, ``Optimal network function virtualization and service function chaining: A survey,'' \emph{Chinese Journal of Electronics}, vol. 27, no. 4, pp. 704-717, September 2018.
	%	\bibitem{Zhang2017}
	%	N. Zhang, Y.-F. Liu, H. Farmanbar, T.-H. Chang, M. Hong, and Z.-Q. Luo, ``Network slicing for service-oriented networks under resource constraints,'' \emph{IEEE Journal on Selected Areas in Communications}, vol. 35, no. 11, pp. 2512-2521, November 2017.
	%	\bibitem{Baumgartner2015}
	%	A. Baumgartner, V. S. Reddy, and T. Bauschert . ``Combined virtual mobile core network function placement and topology optimization with latency bounds,''  in \emph{Proceedings of 4th European Workshop on Software Defined Networks}, 2015, pp. 97-102.
	%	\bibitem{Oljira2017}
	%	D. B. Oljira, K. Grinnemo, J. Taheri, and A. Brunstrom, ``A model for QoS-aware VNF placement and provisioning,'' in  \emph{Proceedings of IEEE Conference on Network Function Virtualization and Software Defined Networks (NFV-SDN)}, 2017, pp. 1-7.
	%	\bibitem{Zhang2019}
	%	N. Zhang, Y.-F. Liu, H. Farmanbar, T.-H. Chang, M. Hong, and Z.-Q. Luo, ``System and method for network slicing for service-oriented networks,'' \emph{US Patent Application} 16/557,169, December 2019.
	%	\bibitem{Liu2017}
	%	J. Liu, W. Lu, F. Zhou, P. Lu, and Z. Zhu, "On dynamic service function chain deployment and readjustment," \emph{IEEE Transactions on Network and Service Management}, vol. 14, no. 3, pp. 543-553, September 2017.
	%	\bibitem{Agyapong2014}
	%	P. K. Agyapong, M. Iwamura, D. Staehle, W. Kiess, and A. Benjebbour, ``Design considerations for a 5G network architecture,'' \emph{IEEE Communications Magazine}, vol. 52, no. 11, pp. 65-75, November 2014.
	%	\bibitem{Domenico2019}
	%	A. De Domenico, Y.-F. Liu, and W. Yu, ``Optimal computational resource allocation and network slicing deployment in 5G hybrid C-RAN,'' in \emph{Proceedings of IEEE International Conference on Communications (ICC)}, 2019, pp. 1-6.
	%	\bibitem{Trivisonno2015}
	%	R. Trivisonno, R. Guerzoni, I. Vaishnavi, and A. Frimpong, ``Network resource management and QoS in SDN-enabled 5G systems,'' in \emph{Proceedings of IEEE Global Communications Conference (GLOBECOM)}, 2015, pp. 1-7.
	%	\bibitem{Luizelli2015}
	%	M. C. Luizelli, L. R. Bays, L. S. Buriol, M. P. Barcellos, and L. P. Gaspary, "Piecing together the NFV provisioning puzzle: Efficient placement and chaining of virtual network functions," in \emph{Proceedings of IFIP/IEEE International Symposium on Integrated Network Management (IM)}, 2015, pp. 98-106.
	%	\bibitem{Jiang2012}
	%	J. W. Jiang, T. Lan, S. Ha, M. Chen, and M. Chiang, ``Joint VM placement and routing for data center traffic engineering,'' in \emph{Proceedings of IEEE INFOCOM}, 2012, pp. 2876-2880.
	%	\bibitem{Guo2011}
	%	T. Guo, N. Wang, K. Moessner, and R. Tafazolli, ``Shared backup network provision for virtual network embedding,'' in \emph{Proceedings of IEEE International Conference on Communications (ICC)}, 2011, pp. 1-5.
	%	\bibitem{Narayana2013}
	%	S. Narayana, W. Jiang, J. Rexford, and M. Chiang, ``Joint server selection and routing for geo-replicated services,'' in \emph{Proceedings of IEEE/ACM 6th International Conference on Utility and Cloud Computing}, 2013, pp. 423-428.
	%	\bibitem{Zhang2015}
	%	S. Q. Zhang, Q. Zhang, H. Bannazadeh, and A. Leon-Garcia, ``Routing algorithms for network function virtualization enabled multicast topology on SDN,'' \emph{IEEE Transactions on Network and Service Management}, vol. 12, no. 4, pp. 580-594, December, 2015.
	%	\bibitem{Addis2015}
	%	B. Addis, D. Belabed, M. Bouet, and S. Secci, ``Virtual network functions placement and routing optimization,'' in \emph{Proceedings of IEEE 4th International Conference on Cloud Networking (CloudNet)}, 2015, pp. 171-177.
	%	\bibitem{Basta2017}
	%	A. Basta, A. Blenk, K. Hoffmann, H. J. Morper, M. Hoffmann, and W. Kellerer, "Towards a cost optimal design for a 5G mobile core network based on SDN and NFV," \emph{IEEE Transactions on Network and Service Management}, vol. 14, no. 4, pp. 1061-1075, December 2017.
	%	\bibitem{Woldeyohannes2018}
	%	Y. T. Woldeyohannes, A. Mohammadkhan, K. K. Ramakrishnan, and Y. Jiang, ``ClusPR: Balancing multiple objectives at scale for NFV resource allocation,'' \emph{IEEE Transactions on Network and Service Management}, vol. 15, no. 4, pp. 1307-1321, December 2018.
	%	\bibitem{Gouareb2018}
	%	R. Gouareb, V. Friderikos, and A. Aghvami, ``Virtual network functions routing and placement for edge cloud latency minimization,''  \emph{IEEE Journal on Selected Areas in Communications}, vol. 36, no. 10, pp. 2346-2357, October 2018.
	%	\bibitem{Xu2013}
	%	H. Xu and B. Li, ``Joint request mapping and response routing for geo-distributed cloud services,'' in \emph{Proceedings of IEEE INFOCOM}, 2013, pp. 854-862.
	%	\bibitem{Li2016}
	%	X. Li, J. B. Rao, and H. Zhang, ``Engineering machine-to-machine traffic in 5G,'' \emph{IEEE Internet of Things Journal}, vol. 3, no. 4, pp. 609-618, August 2016.
	%	\bibitem{Chua2016}
	%	F. C. Chua, J. Ward, Y. Zhang, P. Sharma, and B. A. Huberman, ``Stringer: Balancing latency and resource usage in service function chain provisioning,'' \emph{IEEE Internet Computing}, vol. 20, no. 6, pp. 22-31, November-December 2016.
	%	\bibitem{Charikar2019}
	%	M. Charikar, Y. Naamad, J. Rexford, and X. K. Zou, ``Multi-commodity flow with in-network processing,'' in \emph{Proceedings of International Symposium on Algorithmic Aspects of Cloud Computing (ALGOCLOUD)}, 2019, pp. 73-101.
	%	\bibitem{Carpio2017}
	%	F. Carpio, S. Dhahri and A. Jukan, ``VNF placement with replication for load balancing in NFV networks,'' in \emph{Proceedings of IEEE International Conference on Communications (ICC)},  2017, pp. 1-6
	%	%	\bibitem{Qu2017}
	%	%	L. Qu, C. Assi, K. Shaban, and M. J. Khabbaz, ``A reliability-aware network service chain provisioning with delay guarantees in NFV-enabled enterprise datacenter networks,'' \emph{IEEE Transactions on Network and Service Management}, vol. 14, no. 3, pp. 554–568, 2017.
	%	%	\bibitem{Kammoun2018}
	%	%	A. Kammoun, N. Tabbane, G. Diaz, A. Dandoush, and N. Achir, ``End-to-end efficient heuristic algorithm for 5G network slicing,'' in \emph{Proceedings of IEEE 32nd International Conference on Advanced Information Networking and Applications (AINA)}, 2018, pp. 386–392.
	%	%	\bibitem{Arouk2018}
	%	%	O. Arouk, T. Turletti, N. Nikaein, and K. Obraczka, ``Cost optimization of Cloud-RAN planning and provisioning for 5G networks,'' in \emph{Proceedings of IEEE International Conference on Communications (ICC)}, 2018, pp. 1–6.
	%	\bibitem{Gurobi}
	%	Gurobi Optimization, ``Gurobi optimizer reference manual,'' 2019. [Online]. Available: http://gurobi.com
	%	\bibitem{Mohammadkhan2015}
	%	A. Mohammadkhan, S. Ghapani, G. Liu, W. Zhang, K. K. Ramakrishnan, and T. Wood, ``Virtual function placement and traffic steering in flexible and dynamic software defined networks,'' in \emph{Proceedings of IEEE International Workshop on Local and Metropolitan Area Networks}, 2015, pp. 1-6.
	%	%	\bibitem{Dong2013}
	%	%	J. Dong, X. Jin, H. Wang, Y. Li, P. Zhang, and S. Cheng, ``Energy-saving virtual machine placement in cloud data centers,'' in \emph{Proceedings of IEEE/ACM International Symposium on Cluster, Cloud, and Grid Computing}, 2013, pp. 618-624.
	%	%	\bibitem{Beck2015}
	%	%	M. T. Beck and J. F. Botero, ``Coordinated Allocation of Service Function Chains,'' in \emph{Proceedings of IEEE Global Communications Conference (GLOBECOM)}, 2015, pp. 1-6.
	%	\bibitem{Xu2017}
	%	Y. Xu and V. P. Kafle, ``Reliable service function chain provisioning in software-defined networking,'' in \emph{Proceedings of International Conference on Network and Service Management (CNSM)}, 2017, pp. 1-4.
	%	\bibitem{Cheng2018}
	%	Y. Cheng, L. Yang, and H. Zhu, ``Deployment of service function chain for NFV-enabled network with delay constraint,'' in \emph{Proceedings of International Conference on Electronics Technology (ICET)}, 2018, pp. 383–386.
	%	%			\bibitem{Chen2019}
	%	%			W.-K. Chen, Y.-F. Liu, S. Wang, A. De Domenico, and Z.-Q. Luo, ``A companion technical report of ``A new mathematical formulation of network slicing for service-oriented networks'','' 2019. [Online]. Available: http://lsec.cc.ac.cn/$\sim$yafliu/report\_network\_slicing.pdf.
	%	\bibitem{Gadegaard2018}
	%	{S. L. Gadegaard, A. Klose, and L. R. Nielsen, ``An improved cut-and-solve algorithm for the single-source capacitated facility location problem,'' \emph{EURO Journal Computational Optimization}, vol. 6, no. 1, pp. 1-27, March 2018}.
	%	\bibitem{Guo2019}
	%	Z. Guo, W.-K. Chen, Y.-F. Liu, Y. Xu, and Z.-L. Zhang, ``Joint switch upgrade and controller deployment in hybrid software-defined networks,'' \emph{IEEE Journal on Selected Areas in Communications}, vol. 37, no. 5, pp. 1012-1028, May 2019.
	%	\bibitem{Ahuja1993}
	%	R. K. Ahuja, T. L. Magnanti, and J. B. Orlin, \emph{Network Flows}. Upper Saddle River, NJ, USA: Prentice Hall, 1993.%Prentice Hall, 1993.
	%	\bibitem{Qu2019}
	%	L. Qu, C. Assi, M. Khabbaz, and Y. Ye, ``Reliability-aware service function chaining with function decomposition and multipath Routing,'' \emph{IEEE Transactions on Network and Service Management (to appear)}, 2020. [Online]. Available: https://ieeexplore.ieee.org/abstract/document/8937740/
	%	\bibitem{3gpp}
	%	{3GPP TSG SA5, ``TR 32.972, Telecommunication management, Study on system and functional aspects of energy efficiency in 5G networks,'' Release 16, V. 16.1.0, September 2019.  [Online]. Available: http://www.3gpp.org/DynaReport/32972.htm}
	%	%	\bibitem{Kim2019}
	%	%	K. S. Kim, C. Chae, S. Choi, Y. Ko, J. Kim, Y. Lim, M. Yang, S. Kim, B. Lim, K. Lee, and K. L. Ryu ``Ultrareliable and Low-Latency Communication Techniques for Tactile Internet Services,'' in \emph{Proceedings of the IEEE}, vol. 107, no. 2, pp. 376-393, 2019.
	%	\bibitem{Kaul2012}
	%	S. Kaul, R. Yates, and M. Gruteser, ``Real-time status: How often should one update?'' in \emph{Proceedings IEEE INFOCOM}, 2012, pp. 2731-2735.
	%	\bibitem{Marler2010}
	%	{R. T. Marler and J. S. Arora. ``The weighted sum method for multi-objective optimization: New insights,'' \emph{Structural and Multidisciplinary Optimization}, vol 41, no. 6, pp.  853-862, June 2010.}
	%	\bibitem{Conforti2014}
	%	M. Conforti, G. Cornu\'{e}jols, and G. Zambelli, \emph{Integer Programming}. Cham, Switzerland: Springer, 2014.
	%	\bibitem{Garey1978}
	%		M. R. Garey and D. S. Johnson, ````Strong'' NP-completeness results: Motivation, examples, and implications, '' \emph{Journal of the ACM}, vol. 25, no. 3, pp. 499-508, 1978.
	%	%	\bibitem{Crowder1983}
	%	%	H. Crowder, E. L. Johnson, and M. Padberg, ``Solving large-scale zero-one linear programming problems,'' \emph{Operation Research}, vol. 31, no. 5, pp. 803–834, 1983.
	%	%	\bibitem{Achterberg2016}
	%	%	T. Achterberg, R. E. Bixby, Z. Gu, E. Rothberg, and D. Weninger,
	%	%	``Presolve reductions in mixed integer programming,'' ZIB-Report, 2016, pp. 1–70.
	%	\bibitem{Glover1975}
	%	F. Glover. ``Improved linear integer programming formulations of nonlinear integer problems,''
	%	{\emph{Management Science}}, vol. 22, no. 4, pp. 455-460, December 1975.
	%	\bibitem{Hartman2012}
	%	T. Hartman, A. Hassidim, H. Kaplan, D. Raz, and M. Segalov, ``How to split a flow?'' in \emph{Proceedings of IEEE INFOCOM}, 2012, pp. 828-836.
	%	\bibitem{Dinur2014}
	%	I. Dinur and D. Steurer. ``Analytical approach to parallel repetition'' in \emph{Proceedings of the 46th annual ACM symposium on Theory of computing}, 2014, pp. 624-633.
	%	
	%	%	 \bibitem{bixby2000}
	%	%	 E. R. Bixby, M. Fenelon, Z. Gu, E. Rothberg, and R.Wunderling. ``MIP: Theory and practice--closing the gap'', in \emph{Proceedings of IFIP Conference on System Modeling and Optimization (CSMO)}, 1999, pp. 19-49

\end{thebibliography}
\end{document}